\documentclass[journal]{IEEEtran}
\usepackage{amsmath}
\usepackage{amsfonts}
\usepackage{algorithm}
\usepackage{algorithmic}
\usepackage{algorithm}
\usepackage{array}
\usepackage{textcomp}
\usepackage{stfloats}
\usepackage{url}
\usepackage{verbatim}
\usepackage{graphicx}
\usepackage{cite}
\usepackage{flushend}
\usepackage{amsthm}
\usepackage[T1]{fontenc}
\usepackage{amssymb}
\usepackage{epsfig}
\usepackage{psfrag}
\usepackage{latexsym}
\usepackage{lettrine}
\usepackage{color}
\usepackage{multirow}
\usepackage{subfigure}

\usepackage{bm}
\usepackage{array}
\usepackage{hyperref}					
\hypersetup{colorlinks,
	linkcolor=blue,%
	anchorcolor=blue,
    citecolor=blue}

\newtheorem{thm}{Theorem}
\newtheorem{assumption}{Assumption}
\newtheorem{definition}{Definition}
\newtheorem{lemma}{Lemma}
\newtheorem{corollary}{Corollary}

\newtheorem{proposition}{Proposition}

\begin{document}
\title{On Discrete Ambiguity Functions of Random Communication Waveforms}
\author{
Ying Zhang,~\IEEEmembership{Graduate Student Member,~IEEE}, Fan Liu,~\IEEEmembership{Senior Member,~IEEE},
Yifeng Xiong,~\IEEEmembership{Member,~IEEE}, \\Weijie Yuan,~\IEEEmembership{Senior Member,~IEEE}, Shuangyang Li,~\IEEEmembership{Member,~IEEE}, 
Le Zheng,~\IEEEmembership{Senior Member,~IEEE}, \\Tony Xiao Han,~\IEEEmembership{Senior Member,~IEEE}, Christos Masouros,~\IEEEmembership{Fellow,~IEEE}, and Shi Jin,~\IEEEmembership{Fellow,~IEEE}

\thanks{This work was supported in part by the National Science and Technology Major Projects of China under Grant 2025ZD1302000, and in part by the National Natural Science Foundation of China (NSFC) under Grant 62522107. (\textit{Corresponding authors: Fan Liu; Yifeng Xiong})}
\thanks{Y. Zhang and  W. Yuan are with the School of Automation and Intelligent Manufacturing, Southern University of Science
and Technology, Shenzhen 518055, China. (email: zhangying2024@mail.sustech.edu.cn, yuanwj@sustech.edu.cn).}
\thanks{F. Liu and S. Jin are with the National Mobile Communications Research Laboratory,
Southeast University, Nanjing 210096, China (e-mail: fan.liu@seu.edu.cn, jinshi@seu.edu.cn).}
\thanks{Y. Xiong is with the School of Information and Electronic Engineering,
Beijing University of Posts and Telecommunications, Beijing 100876, China. (e-mail: yifengxiong@bupt.edu.cn)}
\thanks{S. Li is with the Electrical Engineering and Computer Science Department, Technical University of Berlin, 10623 Berlin, Germany (e-mail: shuangyang.li@tu-berlin.de).}
\thanks{L. Zheng is with the Radar Research Laboratory, School of Information and Electronics, Beijing Institute of Technology, Beijing 100081, China (e-mail: le.zheng.cn@gmail.com).}
\thanks{T.-X. Han is with the Wireless Technology Lab, Huawei Technologies Co., Ltd., Shenzhen, 518129, China (e-mail: tony.hanxiao@huawei.com). }
\thanks{C. Masouros is with the Department of Electronic and Electrical Engineering, University College London, London, WC1E 7JE, UK (e-mail: chris.masouros@ieee.org).}
}

\maketitle
\begin{abstract} 
This paper provides a fundamental characterization of the discrete ambiguity functions (AFs) of random communication waveforms under arbitrary orthonormal modulation with random constellation symbols, which serve as a key metric for evaluating the delay-Doppler sensing performance in future ISAC applications. A unified analytical framework is developed for two types of AFs, namely the discrete periodic AF (DP-AF) and the fast-slow time AF (FST-AF), where the latter may be seen as a small-Doppler approximation of the DP-AF. By analyzing the expectation of squared AFs, we derive exact closed-form expressions for both the expected sidelobe level (ESL) and the expected integrated sidelobe level (EISL) under the DP-AF and FST-AF formulations. For the DP-AF, we prove that the normalized EISL is identical for all orthogonal waveforms. To gain structural insights, we introduce a matrix representation based on the finite Weyl–Heisenberg (WH) group, where each delay-Doppler shift corresponds to a WH operator acting on the ISAC signal. This WH-group viewpoint yields sharp geometric constraints on the lowest sidelobes: The minimum ESL can only occur along a one-dimensional cut or over a set of widely dispersed delay-Doppler bins. Consequently, no waveform can attain the minimum ESL over any compact two-dimensional region, leading to a no-optimality (no-go) result under the DP-AF framework. For the FST-AF, the closed-form ESL and EISL expressions reveal a constellation-dependent regime governed by its kurtosis: The Orthogonal Frequency Division Multiplexing (OFDM) modulation achieves the minimum ESL for sub-Gaussian constellations, whereas the Orthogonal Time Frequency Space (OTFS) waveform becomes optimal for super-Gaussian constellations. Finally, four representative waveforms, namely, Single-Carrier (SC), OFDM, OTFS, and Affine Frequency Division Multiplexing (AFDM), are examined under both frameworks, and all theoretical results are verified through numerical examples.



\end{abstract}

\begin{IEEEkeywords} Integrated Sensing and Communication (ISAC), ambiguity function (AF), sidelobe level, Weyl-Heisenberg group.

\end{IEEEkeywords}

\vspace{-0.5em}
\section{Introduction}
\IEEEPARstart{I}{ntegrated} Sensing and Communication (ISAC) has emerged as a transformative design paradigm for the sixth-generation (6G) wireless networks \cite{challengesfor6G2023, 11206742}. In 2023, the International Telecommunication Union (ITU) released its Recommendation Framework for IMT-2030 and identified ISAC as one of the six primary usage scenarios for 6G \cite{ITU2023}. Unlike conventional systems where sensing and communication are treated as independent functionalities, ISAC envisions a tightly coupled integration of the two to support a wide range of emerging applications, including Smart Cities and Homes, Vehicular Networks, and the Low-Altitude Economy \cite{11206742}. Such integration enables the shared use of frequency and power resources, significantly enhancing spectral and energy efficiency \cite{MengkaitaoTWC,zhang2021overview}, while also facilitating mutual performance gains through sensing-assisted communications and communication-assisted sensing \cite{Liufan2023seventyyears}.

One of the fundamental challenges in ISAC is to design waveforms that can simultaneously enable accurate target sensing and maintain reliable communication performance. Existing waveform design strategies are generally grouped into three categories: Sensing-centric, communication-centric, and joint design approaches \cite{7855671, Liufan2018optimalwavedesign}. Sensing-centric methods embed communication data into conventional radar signals, while communication-centric schemes repurpose standardized communication waveforms to perform sensing tasks \cite{R.M.M_1963, MR_2003, 8424569, Sturm2011}. Joint design approaches construct novel ISAC waveforms with the objective of achieving a balanced tradeoff between sensing accuracy and communication efficiency \cite{liu-2021-cramer, 9924202, 9724205}. Owing to their inherent compatibility with current cellular protocols, communication-centric designs can be seamlessly integrated into existing network infrastructures without significant hardware modifications, thereby enabling practical ISAC deployment with minimal impact on established communication performance.

In current communication-centric ISAC designs, sensing functionality is largely supported by a small set of physical layer reference signals, such as channel state information reference signals (CSI-RS) and sounding reference signals (SRS) \cite{9805471,10904093,10630588}. These signaling elements are deterministic and well structured, which provides desirable auto-and cross-correlation characteristics for delay-Doppler estimation tasks. Nevertheless, they typically account for only around 10\% of the entire frame, leaving the majority of resources unused from the sensing perspective. Given the increasingly stringent sensing accuracy requirements foreseen for 6G networks, relying solely on these signals is no longer adequate. It is therefore natural to consider exploiting the remaining 90\% time-frequency resources, which primarily transmit communication data payloads, in order to reinforce sensing capability and improve overall resource efficiency. This strategy, however, introduces a fundamental conflict: While deterministic signals are ideal for radar sensing, data payloads are inherently random, and such randomness can significantly degrade sensing performance \cite{LushihangNetwork, LushihangTSP}. This critical issue has recently been identified as the {\it{deterministic-random tradeoff (DRT)}} in ISAC systems \cite{xiong2023fundamental, Xiong2024Magazine}. Understanding how sensing behaves under such randomness is therefore essential. Accordingly, a careful theoretical analysis is required to evaluate the sensing capability when the waveform carries random data.

The ambiguity function (AF), first introduced by Woodward in 1953 \cite{woodward1953probability}, characterizes the two-dimensional (2D) self-convolution of a deterministic radar waveform in the delay-Doppler domain. For periodic signals, or signals equipped with a cyclic prefix (CP), the AF can be further generalized to the periodic AF \cite{303781,144564}. As a foundational tool in radar sensing, the AF enables the evaluation of two key performance metrics: delay-Doppler resolution, governed by the mainlobe width, and mutual interference among targets, determined by the sidelobe level. Consequently, shaping the AF or its zero-Doppler slice (also known as the auto-correlation function (ACF)) to achieve desirable ambiguity characteristics has been a longstanding research topic since the early development of radar systems \cite{11003993,10935731,1056980}. Beyond waveform design, the AF is also closely linked to fundamental estimation-theoretic limits. Under a single-target scenario with additive white Gaussian noise (AWGN), maximum likelihood estimation (MLE) reduces to identifying the peak of the AF \cite{kay1998fundamentals1}. Furthermore, the curvature of the AF mainlobe is proportional to the Fisher information and thus directly linked to the Cram\'er–Rao bound (CRB) \cite{van2004detection}. In addition, the Ziv–Zakai bound (ZZB), widely recognized as a tighter bound than the CRB \cite{556118,10884827}, can be expressed as a nonlinear functional of the AF in the single-target AWGN case \cite{6894212}.

In addition to Woodward’s conventional AF definition, the fast-slow-time (FST) ambiguity function (FST-AF) also plays a critical role in sensing performance assessment. Specifically, the fast-time dimension corresponds to the samples within a single radar pulse and is used for delay estimation by measuring their arrival times. The slow-time dimension, in contrast, spans consecutive pulses and enables Doppler estimation through phase variations across pulses \cite{9492001, 10242003, 8793214}. In radar signal processing, the FST representation typically aligns with the widely adopted ``stop-and-go'' model, where the Doppler shift is assumed constant within each pulse and only varies across pulses \cite{Richards2005Fundamentals}. This approximation is valid in scenarios with small or moderate Doppler frequencies, where inter-carrier interference (ICI) is negligible. The FST-AF is defined as the two-dimensional discrete Fourier transform (2D-DFT) of the element-wise squared signal in the frequency–slow-time domain. At the receiver, the 2D-DFT operation over the echo signal yields the multi-target delay-Doppler profile, which may be interpreted as a noisy linear combination of multiple time-frequency shifted FST-AFs.

While the AF theory has been extensively developed for deterministic radar waveforms, its direct application to communication-centric ISAC signals remains challenging. Due to the inherent randomness of communication data symbols, the resulting AF becomes a random function, whose statistical characteristics strongly depend on the underlying modulation waveform and constellation alphabet. Consequently, assessing the mainlobe and sidelobe behavior based on a single realization of data sequence is no longer meaningful. Instead, the AF must be evaluated in an statistical sense, which forms the central motivation of this work. 

In our previous study \cite{liu2024ofdmachieveslowestranging}, we derived closed-form expressions for the average squared ACF of random communication signals under both periodic and aperiodic configurations. The results rigorously proved that, among all orthogonal waveforms employing Quadrature Amplitude Modulation (QAM) or Phase Shift Keying (PSK) constellations, OFDM achieves the lowest ranging sidelobe level, thereby enabling superior ranging accuracy. These findings were further extended to pulse-shaped signals in \cite{11037613}, where a metaphorical ``iceberg hidden in the sea'' structure was introduced to describe the statistical behavior of the ACF. Specifically, the ``iceberg'' corresponds to the squared mean of the ACF, determined by the squared ACF of the adopted pulse-shaping filter, whereas the ``sea level'', arising from the randomness of data payloads, represents the variance of the ACF. However, these studies were restricted to the zero-Doppler characteristics of the AF. A complete characterization of the ambiguity properties of random communication signals in the full delay-Doppler domain therefore remains largely unexplored. This necessitates a deeper investigation into the statistical ambiguity properties of communication-centric ISAC waveforms, beyond conventional deterministic waveform analysis for radar systems.

To bridge this gap, this work focuses on analyzing the discrete AFs of orthogonal communication waveforms. Specifically, we aim to address the following two fundamental questions:
\begin{itemize}
    \item How to evaluate the average delay-Doppler sidelobe level of communication waveforms carrying random symbols?
    \item What is the optimal communication-centric waveform for joint delay and Doppler estimation?
\end{itemize}

To answer the first question, we investigate both the discrete periodic AF (DP-AF) and the FST-AF of communication waveforms with independently and identically distributed (i.i.d.) constellation symbols. In particular, we derive generic closed-form expressions for the average squared DP-AF and FST-AF under arbitrary modulation bases. We further examine several representative waveform designs commonly employed in communication systems, including SC, OFDM, OTFS, and AFDM.

To answer the second question, we introduce the expected sidelobe level (ESL) and the expected integrated sidelobe level (EISL) as sensing performance metrics for random communication waveforms. We show that the normalized EISL of the DP-AF remains unchanged across different modulation schemes and constellations. This is a direct consequence of the volume invariance of classical AFs. In addition, we prove that no orthogonal waveform can achieve the minimum ESL over any compact two-dimensional region in the delay-Doppler plane. At best, the lowest sidelobes may occur either along a one-dimensional cut, such as zero-delay or zero-Doppler cuts, or at a set of well-separated delay-Doppler bins. Therefore, there does not exist any waveform that is globally or locally optimal for minimizing the ESL in the DP-AF framework. In contrast, under the FST-AF formulation, OFDM achieves the minimum EISL and average sidelobe level when using sub-Gaussian constellations such as QAM and PSK, while OTFS attains the minimum values when employing super-Gaussian constellations.

The remainder of this paper is organized as follows. Sec. \ref{sec_2} introduces the system model and formally defines the ESL and EISL metrics for both AF formulations. Secs. \ref{sec_3} and \ref{sec_4} analyze the DP-AF and FST-AF, respectively, and further investigate several representative communication waveforms. Finally, Sec. \ref{sec_5} concludes the paper.

{\emph{Notations}}: Matrices are denoted by bold uppercase letters (e.g., $\mathbf{U}$), vectors are represented by bold lowercase letters (e.g., $\mathbf{x}$), and scalars are denoted by normal font (e.g., $N$); The $n$th entry of a vector $\mathbf{s}$, and the $(m,n)$-th entry of a matrix $\mathbf{A}$ are denoted as $s_n$ and $a_{m.n}$, respectively; $\otimes$, $\odot$ and $\operatorname{vec}\left(\cdot\right)$ denote the Kronecker product, the Hadamard product and the vectorization, $\left(\cdot\right)^T$, $\left(\cdot\right)^H$, and $\left(\cdot\right)^\ast$ stand for transpose, Hermitian transpose, and the complex conjugate of the matrix; The entry-wise square of a matrix $\mathbf{X}$ is denoted as $\left| \mathbf{X} \right|^2 = \mathbf{X} \odot \mathbf{X}^\ast$; The $\ell_p$ norm and Frobenius norm are written as $\left\| \cdot\right\|_p$ and $\left\| \cdot\right\|_F$, and $\mathbb{E}(\cdot)$ represents the expectation operation; $\circledast$ denotes the circular convolution; The notation $ \mathrm{Diag}(\mathbf{a})$ denotes the diagonal matrix obtained by placing the entries of $\mathbf{a}$ on its main diagonal; $\mathbf{1}_{N,N}$ and $\mathbf{1}_{N}$ represent the all-one matrix with size $N \times N$ and the all-one vector with length $N$; $\mathrm{Card}\left(\cdot\right)$ denotes the cardinality of a set; $\langle \cdot \rangle_N$ denotes the modulo $N$ operation; The symbol $\delta_{m,n}$ denotes the Kronecker-Delta function:
\begin{equation*}
    \delta_{m,n} = \left\{ 
    \begin{aligned}
        & 0, \quad m \ne n; \\
        & 1, \quad m=n.
    \end{aligned}
    \right.
\end{equation*}

\section{System Model and Performance Metrics}\label{sec_2}

\subsection{Communication-Centric ISAC Signal Model}

In this section, we introduce two types of communication-centric ISAC signal models, namely, the one-dimensional (1D) and two-dimensional (2D) models. While the 1D signaling naturally models the consecutive transmission of the time-domain signal, in practical communication systems, multiple constellation symbols are often grouped and mapped into two dimensions, thereby forming a 2D signal model. Specifically, we adopt the DP-AF and FST-AF to analyze the 1D and 2D models, respectively.

\subsubsection{1D Signaling}
Let us consider the transmission of $N$ i.i.d. communication symbols $\mathbf{s} = \left[s_1, s_2, \dots, s_N \right]^T \in \mathbb{C}^{N}$, which are drawn from a proper constellation $\mathcal{S}$, and modulated over an orthonormal basis in the time domain. Such a basis can be represented by a unitary matrix $\mathbf{U} = [\mathbf{u}_1, \mathbf{u}_2,...,\mathbf{u}_N] \in \mathcal{U}(N)$, where $\mathcal{U}(N)$ denotes the unitary group of degree $N$. Consequently, the discrete time-domian signal can be given by
\begin{equation} \label{eq:1-D signal}
    \mathbf{x} = \mathbf{U} \mathbf{s} = \sum_{n=1}^N s_n \mathbf{u}_n\in\mathbb{C}^N.
\end{equation}
This general model captures a wide range of practical communication signaling schemes, including four representative examples: SC, OFDM, OTFS, and AFDM, given as
\begin{subequations}\label{1D signal of four typical signals}
\begin{align}
        &\mathbf{x}_{\rm{SC}} = \mathbf{I}_N \cdot \mathbf{s},\\
        &\mathbf{x}_{\rm{OFDM}} = \mathbf{F}_N^H \cdot \mathbf{s}, \\
        &\mathbf{x}_{\rm{OTFS}} = \left( \mathbf{F}_{N_1}^H \otimes \mathbf{I}_{N_2} \right) \cdot \mathbf{s}, \quad {N_1}{N_2} = N,\\
        &\mathbf{x}_{\rm{AFDM}} = \left( \mathbf{\Lambda}_{c_1}^H \mathbf{F}_N^H \mathbf{\Lambda}_{c_2}^H \right) \cdot \mathbf{s},
\end{align}
\end{subequations}
where $\mathbf{F}_N$ is the normalized discrete Fourier transform (DFT) matrix of size $N$, with its $(m,n)$-th entry being defined as $\frac{1}{\sqrt{N}}e^{-\frac{j2\pi(m-1)(n-1)}{N}}$, and
\begin{align}
    &\mathbf{\Lambda}_{c} = \mathrm{Diag} \left( \left[ 1, e^{-j2\pi c}, ...,e^{-j2\pi c (N-1)^2} \right]^T\right):=\mathrm{Diag} \left( \mathbf{c}\right),
\end{align}
where $c$ can be either $c_1$ or $c_2$, and $2Nc_1\in\mathbb{Z}$.

To proceed, let us further impose some generic constraints on the considered constellation symbols.
\begin{assumption}[Unit Power and Rotational Symmetry]
    We focus on constellations with unit power, zero mean, and zero pseudo-variance, defined as
    \begin{equation}
        \mathbb{E}( \left|s \right|^2) = 1, \quad\mathbb{E}(s) = 0,\quad\mathbb{E}(s^2) = 0,\quad s\in\mathcal{S}.
    \end{equation}
\end{assumption}

\noindent The unit-power normalization in Assumption 1 enables fair comparison of sensing and communication performance under varying constellation formats. Moreover, most practical constellations, such as the PSK/QAM families, satisfy the zero mean and zero pseudo-variance properties, whereas BPSK and 8-QAM are two outliers that do not satisfy the latter.

We now define the \textit{kurtosis} of a constellation, which serves as a key component in shaping the AF of communication signals, and is given by
\begin{equation}
    \kappa = \frac{\mathbb{E}\left\{|s-\mathbb{E}(s)|^4\right\}}{\mathbb{E}^2\left\{|s-\mathbb{E}(s)|^2\right\}}.
\end{equation}
Under unit power and zero mean conditions, kurtosis reduces to the fourth-order moment $\kappa= \mathbb{E}(|s|^4)$. Using the power mean inequality yields
\begin{equation}
    \mathbb{E} \left( \left|s \right|^4 \right) \geq \mathbb{E}^2\left( \left|s\right|^2\right) = 1.
\end{equation}
The standard complex Gaussian constellation, which has a kurtosis of 2, also satisfies the above criteria. Taking the Gaussian constellation as the baseline, we may define the following two categories of constellations.
\begin{definition}
    [Sub-Gaussian Constellation] A constellation is sub-Gaussian when its kurtosis is less than 2, subject to Assumption 1.
\end{definition}
\begin{definition}
    [Super-Gaussian Constellation] A constellation is sub-Gaussian when its kurtosis is greater than 2, subject to Assumption 1.
\end{definition}
\noindent Specifically, PSK exhibit $\kappa = 1$, while all QAM constellations have $1\leq \kappa < 2$, both of which are sub-Gaussian.

\subsubsection{2D Signaling}
Let us consider a generic 2D modulation scheme, which may be represented as an $N\times M$ time-domain signal matrix as
\begin{equation} \label{eq:2-D signal}
    \mathbf{X} = \mathbf{U} \mathbf{S} \mathbf{V}\in\mathbb{C}^{N\times M},
\end{equation}
where $\mathbf{S}\in\mathbb{C}^{N\times M}$ contains i.i.d. constellation symbols, and $\mathbf{U}\in\mathcal{U}\left(N\right)$ and $\mathbf{V}\in\mathcal{U}\left(M\right)$ are unitary matrices with dimensions $N$ and $M$, respectively. Without the loss of generality, we name each row and column of $\mathbf{X}$ as a fast-time and a slow-time sample, respectively. Accordingly, $\mathbf{U}$ and $\mathbf{V}$ may be respectively treated as fast- and slow-time signaling bases. The fast-time domain is sometimes referred to as the delay domain, whereas the slow-time domain is the Fourier dual of the Doppler domain. Note that in practical systems, the time-domain signal is transmitted in a serial manner, i.e.,
\begin{equation}
    \operatorname{vec}\left(\mathbf{X}\right) = \left(\mathbf{V}^T\otimes\mathbf{U}\right)\operatorname{vec}\left(\mathbf{S}\right)\in\mathbb{C}^{NM\times 1},
\end{equation}
which may be realized through simple parallel-to-serial transform over $\mathbf{X}$.
Accordingly, the 2D counterparts of the four typical modulation schemes in \eqref{1D signal of four typical signals} may be formulated as
\begin{subequations}
\begin{align}\label{2D signal of four typical signals}
    &\mathbf{X}_{\rm SC} = \mathbf{I}_N \cdot \mathbf{S} \cdot \mathbf{I}_M, \\
    &\mathbf{X}_{\rm OFDM} = \mathbf{F}_N^H\cdot\mathbf{S}\cdot\mathbf{I}_M,\\
    &\mathbf{X}_{\rm OTFS} = \mathbf{I}_N\cdot\mathbf{S}\cdot\mathbf{F}_M^H,\\
    &\mathbf{X}_{\rm AFDM} = \left(\mathbf{\Lambda}_{c_1}^{H} \mathbf{F}_N^H \mathbf{\Lambda}_{c_2}^H \right)\cdot\mathbf{S}\cdot\mathbf{I}_M,
\end{align}
\end{subequations}
both of which occupy $N$ fast-time and $M$ slow-time slots.

\subsection{Sensing Signal Processing and Performance Metrics}
\subsubsection{Preliminaries}
Before proceeding with the technical development, we introduce several matrix operators and key concepts that will be used throughout the analysis. Consider a discrete-time signal of length $N$. A time shift by $k$ samples and a Doppler frequency shift indexed by $q$ can be represented via matrix multiplications. In particular, the aperiodic (non-circular) time-shift matrix is denoted as
\begin{align}
    \tilde{\mathbf{J}}_{N,k} := \left[ 
    \begin{matrix}
        \mathbf{0} &\mathbf{0}\\
        \mathbf{I}_{N-k} & \mathbf{0}
    \end{matrix}
    \right]\in\mathbb{R}^{N\times N},
\end{align}
which implements a downward shift by $k$ samples with zero-padding.
Moreover, the frequency-shift matrix is defined by
\begin{align}
    \mathbf{D}_{N,q}:&\nonumber =\mathrm{Diag}\left(1,e^{j\frac{2\pi q}{N}},\ldots,e^{j\frac{2\pi q(N-1)}{N}}\right) \\
    &= \sqrt{N}\mathrm{Diag}(\mathbf{f}_{N,q+1}^\ast),
\end{align}
where $\mathbf{f}_{N,q+1}$ is the $q$-the column of the size-$N$ DFT matrix. 

In multi-carrier modulation systems, a CP is typically appended at the transmitter to mitigate delay spread caused by multipath propagation and to maintain subcarrier orthogonality. Let the CP length be $N_{\text{CP}}\leq N$. The CP may be appended to a length-$N$ signal by multiplying with:
\begin{equation}
    \mathbf{A}_{\rm{CP}}: = \left[ 
    \begin{matrix}
        \mathbf{0}_{ N_{\text{CP}} \times (N-N_{\text{CP}}) }, \mathbf{I}_{N_{\text{CP}}}\\
         \mathbf{I}_N
    \end{matrix}
    \right] \in \mathbb{R}^{(N+N_{\text{CP}})\times N}.
\end{equation}
After reception, the CP is discarded via the CP-removal matrix
\begin{equation}
    \mathbf{R}_{\rm{CP}}:= \left[ \begin{matrix}
        \mathbf{0}_{N \times N_{\text{CP}}},\mathbf{I}_N
    \end{matrix} \right] \in \mathbb{R}^{N\times(N+N_{\text{CP}})}.
\end{equation}
A well-known identity among the above matrices is
\begin{equation}
   \mathbf{R}_{\rm{CP}} \tilde{\mathbf{J}}_{{N+N_{\text{CP}},k}}\mathbf{A}_{\rm{CP}}  = \left[ \begin{matrix}
        \mathbf{0} & \mathbf{I}_k \\
        \mathbf{I}_{N-k} & \mathbf{0}
    \end{matrix}\right]:=\mathbf{J}_{N,k},
\end{equation}
where $\mathbf{J}_{N,k}\in\mathbb{R}^{N\times N}$ denotes the periodic time-shift matrix. This identity shows that CP insertion and removal effectively convert an aperiodic time shift into its periodic counterpart. 

Define $\mathbf{G}_{k,q}:=\mathbf{D}_{N,q}\mathbf{J}_{N,k}$, where $k,q\in\mathbb{Z}_N$. Here, $\mathbb{Z}_N$ denotes the cyclic group of integers modulo $N$, i.e., $\mathbb{Z}_N:=\left\{0,1,\ldots,N-1\right\}$, with all operations performed modulo $N$. It can be shown that the set  
\begin{equation}
    \mathcal{H}\left(N\right):=\left\{e^{j\ell\frac{2\pi}{N}}\mathbf{G}_{k,q}\right\}_{(\ell,k,q)\in\mathbb{Z}_N^3}
\end{equation} 
constitutes a unitary representation of the finite dimensional Weyl–Heisenberg (WH) group. A fundamental property of the WH group is the non-commutation relation between time- and frequency-shifts, which is
\begin{align}\label{non_commutation_property}
    \mathbf{D}_{N,q} \mathbf{J}_{N,k} = e^{-\frac{j 2\pi qk}{N}} \mathbf{J}_{N,k} \mathbf{D}_{N,q}.
\end{align}
With these matrix definitions in place, we now proceed to formally derive the DP-AF and FST-AF.

\subsubsection{DP-AF}
Let us commence with the simplest case where a time-domain signal $\mathbf{x} = \mathbf{U}\mathbf{s}\in\mathbb{C}^{N}$ is appended with a CP, and is then propagated through a single-target channel with delay-Doppler indices $k_0,\tilde{q}_0$, as well as the complex channel gain $\tilde{\beta}_0\in\mathbb{C}$. The noiseless version of the received signal is expressed as
\begin{equation}
    \tilde{\mathbf{x}}  = \tilde{\beta}_0\mathbf{D}_{{N+N_{\text{CP}},\tilde{q}}_0}\tilde{\mathbf{J}}_{{N+N_{\text{CP}},k_0}}\mathbf{A}_{\rm{CP}} \mathbf{x}.
\end{equation}
The CP-removal operation at the receiver yields
\begin{align}
    &\nonumber \mathbf{R}_{\rm{CP}}\tilde{\mathbf{x}} = \tilde{\beta}_0\mathbf{R}_{\rm{CP}}\mathbf{D}_{{N+N_{\text{CP}},\tilde{q}_0}}\tilde{\mathbf{J}}_{{N+N_{\text{CP}},k_0}}\mathbf{A}_{\rm{CP}} \mathbf{x}\\
    &\nonumber = \tilde{\beta}_0\mathrm{Diag}\left(e^{j \frac{2 \pi\tilde{q}_0 N_{\rm{CP}}}{N+N_{\text{CP}}}},..., e^{j \frac{2 \pi\tilde{q}_0 (N+N_{\rm{CP}}-1)}{N + N_{\text{CP}}}}\right)
        \mathbf{J}_{N,k_0} \mathbf{x}\\
    &= \tilde{\beta}_0e^{j \frac{2 \pi\tilde{q}_0 N_{\rm{CP}}}{N+N_{\text{CP}}}} \mathrm{Diag}\left(1,..., e^{j \frac{2 \pi\tilde{q}_0 (N-1)}{N + N_{\text{CP}}}}\right)
        \mathbf{J}_{N,k_0} \mathbf{x}.
\end{align}
By letting $\beta_0 = \tilde{\beta}_0e^{j \frac{2 \pi \tilde{q}_0 N_{\rm{CP}}}{N+N_{\text{CP}}}}$, $q_0=\frac{N\tilde{q}_0}{N+N_{\text{CP}}}$, \footnote{Note that in practical signal processing, such equivalence can be established through down-sampling operation.} we have
\begin{equation}
    \mathbf{R}_{\rm{CP}}\tilde{\mathbf{x}}=\beta_0\mathbf{D}_{N,q}\mathbf{J}_{N,k}\mathbf{x}.
\end{equation}

The above single-path noiseless case may be readily extended to a generic scenario where $L$ targets are located at different ranges and moving with different velocities. Each target is characterized by a complex amplitude $\beta_\ell$, a delay $\tau_\ell$, and a Doppler $\nu_\ell$, $\ell = 1,2,\ldots,L$. After CP removal at the receiver, the received echo signal may be given as
\begin{equation}
     \mathbf{y} = \sum_{\ell=1}^L \beta_\ell \mathbf{D}_{N,q_\ell}\mathbf{J}_{N,k_\ell}\mathbf{x} + \mathbf{z},
\end{equation}
where $\mathbf{z}$ denotes the AWGN. Upon receiving $\mathbf{y}$, a common practice is to matched-filter (MF) $\mathbf{y}$ by a time-frequency shifted counterpart of $\mathbf{x}$, yielding the output:
\begin{align}
    &\nonumber \tilde{y}_{k,q}^{\rm{MF}} = \mathbf{x}^H \mathbf{J}_{N,k}^T \mathbf{D}_{N,q}^\ast \mathbf{y} \\
    &\nonumber = \sum_{\ell=1}^L\beta_\ell \mathbf{x}^H \mathbf{J}_{N,k}^T \mathbf{D}_{N,q}^\ast\mathbf{D}_{N,q_\ell}\mathbf{J}_{N,k_\ell}\mathbf{x} + \tilde{z}_{k,q}\\
    &\nonumber = \sum_{\ell=1}^L\beta_\ell \mathbf{x}^H \mathbf{J}_{N,k}^T \mathbf{D}_{N,{q-q_\ell}}^\ast\mathbf{J}_{N,k_\ell}\mathbf{x}+ \tilde{z}_{k,q}\\
    &\nonumber = \sum_{\ell=1}^L\beta_\ell e^{\frac{j 2\pi k_\ell(q-q_{\ell})}{N}} \mathbf{x}^H \mathbf{J}_{N,k}^T \mathbf{J}_{N,k_\ell}\mathbf{D}_{N,{q-q_\ell}}^\ast\mathbf{x}+ \tilde{z}_{k,q}\\
    &=\sum_{\ell=1}^L\beta_\ell e^{\frac{j 2\pi k_\ell(q-q_{\ell})}{N}} \mathbf{x}^H \mathbf{J}_{N,k-k_\ell}^T \mathbf{D}_{N,{q-q_\ell}}^\ast\mathbf{x}+ \tilde{z}_{k,q},
\end{align}
where $\tilde{z}_{k,q} = \mathbf{x}^H \mathbf{J}_{N,k}^T \mathbf{D}_{N,q}^\ast\mathbf{z}$. The DP-AF may be accordingly defined as
\begin{align}
    \mathcal{A}_{\rm{DP}}(k,q) 
    &=: \mathbf{x}^H \mathbf{J}_{N,k}^T \mathbf{D}_{N,q}^\ast \mathbf{x}, \quad \left(k,q\right)\in\mathbb{Z}_N^2,
\end{align}
under which the MF output can be rewritten as
\begin{equation}
    \tilde{y}_{k,q}^{\rm{MF}} =\sum_{\ell=1}^L\beta_\ell e^{\frac{j 2\pi k_\ell(q-q_{\ell})}{N}} \mathcal{A}_{\rm{DP}}(k-k_\ell,q-q_\ell)+ \tilde{z}_{k,q},
\end{equation}
which is a linear combination of time-frequency shifted DP-AFs. In order to facilitate multi-target detection, it is desirable that the squared output $|\tilde{y}_{k,q}^{\rm{MF}}|^2$ exhibits pronounced peaks at $(k,q)=(k_\ell,q_\ell)$ and low sidelobe levels elsewhere. This is strongly influenced by the overall structure of $\mathcal{A}_{\rm{DP}}(k,q)$, motivating the definition of sidelobe level of $\mathcal{A}_{\rm{DP}}(k,q)$ as
\begin{align}\label{2D AF}
    \nonumber |\mathcal{A}_{\rm{DP}}(k,q)|^2 &= |\mathbf{x}^H \mathbf{J}_{N,k}^T \mathbf{D}_{N,q}^\ast \mathbf{x}|^2, \\
    & = |\mathbf{x}^H\mathbf{D}_{N,q}\mathbf{J}_{N,k} \mathbf{x}|^2,\quad \forall (k,q) \ne (0,0).
\end{align}
The mainlobe of DP-AF is $|\mathcal{A}_{\rm{DP}}(0,0)|^2 = |\mathbf{x}^H \mathbf{x}|^2 = \|\mathbf{x}\|_2^4$. It follows that the integrated sidelobe level (ISL) can be expressed as
\begin{align}
    {\rm{ISL}}_{\text{DP}} = \sum_{k=0}^{N-1} \sum_{q=0}^{N-1} |{\mathcal{A}_{\rm{DP}}}(k,q)|^2 - |\mathcal{A}_{\rm{DP}}(0,0)|^2.
\end{align}
Due to the random nature of the signal, the DP-AF is a random function. Therefore, we define the ESL as a sensing metric to characterize the average performance, which is expressed as
\begin{align} \label{the expected sidelobe level}
    & \mathbb{E}(\left| \mathcal{A}_{\rm{DP}}(k,q) \right|^2)= \mathbb{E}(|\mathbf{x}^H \mathbf{D}_{N,q}\mathbf{J}_{N,k}  \mathbf{x}|^2), \;\;(k,q)\ne (0,0),
\end{align}
where the expectation is taken over the random symbol vector $\mathbf{s}$. In sight of $|\mathcal{A}_{\rm{DP}}(0,0)|^2  = \|\mathbf{x}\|_2^4$, the average mainlobe level may be computed by \cite{liu2024ofdmachieveslowestranging} 
    \begin{equation}\label{the average mainlobe of AF}
        \mathbb{E}(\left| \mathcal{A}_{\rm{DP}}(0,0) \right|^2) = \mathbb{E}(\|\mathbf{x}\|_2^4) =N^2 + (\kappa - 1)N,
    \end{equation}
and the corresponding EISL is given by
\begin{equation}\label{EISL-DP}
\begin{aligned}
    \mathrm{EISL}_{\rm{DP}} 
    & =\sum_{k=0}^{N-1} \sum_{q=0}^{N-1} \mathbb{E}(|\mathcal{A}_{\rm{DP}}(k,q)|^2) - \mathbb{E}(|\mathcal{A}_{\rm{DP}}(0,0)|^2).
\end{aligned}
\end{equation}

\subsubsection{FST-AF}
Let us consider the transmission of a size-$MN$ ISAC signal $\mathbf{x} = \operatorname{vec}\left(\mathbf{X}\right)\in\mathbb{C}^{MN\times 1}$. By assuming that the Doppler frequency is sufficiently small, the resulting phase shift can be considered constant over blocks of $N$ samples, changing only between them. This allows the signal $\mathbf{x}$ to be partitioned into $M$ column vectors, forming the 2D FST representation $\mathbf{X} = \left[\mathbf{x}_1, \mathbf{x}_2, ..., \mathbf{x}_M \right]\in\mathbb{C}^{N\times M}$. At the transmitter, a CP is added to each slow-time block (each column) to cover the maximum target delay. Suppose the length of CP is $N_{\rm{CP}}$, the signal with CP can be given by
\begin{align}
    \mathbf{X}_{\rm{CP}} = \mathbf{A}_{\rm{CP}} \mathbf{X} \in\mathbb{C}^{(N+N_{\text{CP}})\times M}.
\end{align}
To transmit the signal in the physical channel, a parallel-to-serial transform is performed over $\mathbf{X}_{\rm{CP}}$, leading to
\begin{align}\label{FST_add_CP}
    \mathbf{x}_{\rm{CP}} = \operatorname{vec} (\mathbf{X}_{\rm{CP}}) = (\mathbf{I}_M \otimes \mathbf{A}_{\rm{CP}}) \mathbf{x}.
\end{align}
For simplicity, we initially focus on a single-target scenario, with delay-Doppler indices $k_0,q_0$ and channel gain $\tilde{\beta}_0\in\mathbb{C}$. After propagating through the physical channel, the received noiseless signal is
\begin{align}
    \tilde{\mathbf{x}} = \tilde{\beta}_{0} \mathbf{D}_{M(N+N_{\rm{CP}}), q_0}\tilde{\mathbf{J}}_{M(N+N_{\rm{CP}}), k_0} \mathbf{x}_{\rm{CP}},
\end{align}
where $k_0 < N_{\rm{CP}}$. Assuming that the fast-time Doppler is sufficiently small, such that $\frac{2\pi q_0N}{M(N+N_{\rm{CP}})} \ll 1$, one may employ the approximation
\begin{align}
    &e^{\frac{j2\pi q_0\left[\left(m-1\right)(N+N_{\rm{CP}})+(n-1)\right]}{M(N+N_{\rm{CP}})}} 
    \nonumber = e^{\frac{j2\pi q_0\left(m-1\right)(N+N_{\rm{CP}})}{M(N+N_{\rm{CP}})}}e^{\frac{j2\pi q_0 (n-1)}{M(N+N_{\rm{CP}})}}\\
    &\approx e^{\frac{j2\pi q_0(m-1)}{M}}, \;\; m = 1,2,\ldots,M, \;\; n = 1,2,\ldots,N+N_{\rm CP}.
\end{align}
This approximation implies that each fast-time sample within the $m$-th slow-time block experiences the same Doppler phase shift, aligning with the ``stop-and-go'' model in pulse-Doppler radar signal processing. \footnote{Note that this assumption may still hold in large-Doppler scenarios when the CP is replaced by zero padding (ZP). In this case, the ZP guard interval can be extended as much as needed, whereas the CP length is inherently limited to at most $N$. A sufficiently long ZP suppresses the Doppler-induced phase variation within each fast-time symbol, which is analogous to a pulsed radar operating with a small duty cycle.} This suggests that the Doppler shift matrix can be approximated as
\begin{align}
    \mathbf{D}_{M(N+N_{\rm{CP}}), q_0} \approx \mathbf{D}_{M, q_0} \otimes \mathbf{I}_{N+N_{\rm{CP}}},
\end{align}
up to a global phase shift. Consequently, the received signal can be approximated as
\begin{align}
    \tilde{\mathbf{x}} \approx  \beta_0\left(\mathbf{D}_{M, q_0}\otimes \mathbf{I}_{N+N_{\rm{CP}}}\right) \tilde{\mathbf{J}}_{M(N+N_{\rm{CP}}), k_0} \mathbf{x}_{\rm{CP}},
\end{align}
where the global phase shift is incorporated in the channel gain $\beta_0$. By recalling \eqref{FST_add_CP}, and using the identity
\begin{equation}
    \left( \mathbf{I}_M \otimes \mathbf{R}_{\rm{CP}} \right)  \tilde{\mathbf{J}}_{M(N+N_{\rm{CP}}),k_0}(\mathbf{I}_M \otimes \mathbf{A}_{\rm{CP}}) = \mathbf{I}_M\otimes\mathbf{J}_{N,k_0},
\end{equation}
the received signal after CP removal can be simplified to
\begin{align}
    \mathbf{y}  = \left( \mathbf{I}_M \otimes \mathbf{R}_{\rm{CP}} \right) \tilde{\mathbf{x}} \approx  \beta_0\left( \mathbf{D}_{M, q_0} \otimes  \mathbf{J}_{N,k_0} \right) \mathbf{x}.
\end{align}
Similarly, in the $L$-target scenario, the received signal becomes
\begin{equation}
    \mathbf{y} \approx \sum_{\ell=1}^L \beta_\ell \left(\mathbf{D}_{N,q_\ell}\otimes\mathbf{J}_{N,k_\ell}\right)\mathbf{x} + \mathbf{z}.
\end{equation}

By performing MF over $\mathbf{y}$, we have
\begin{align}
    &\nonumber \tilde{y}_{k,q}^{\rm{MF}} = \mathbf{x}^H \left(\mathbf{D}_{M,q}^\ast\otimes\mathbf{J}_{N,k}^T\right)\mathbf{y}  \\ 
    & \approx  \sum_{\ell = 1}^L\beta_\ell\mathbf{x}^H\left(\mathbf{D}_{M,q-q_\ell}^\ast\otimes\mathbf{J}_{N,k-k_\ell}^T\right)\mathbf{x} + \tilde{z}_{k,q},
\end{align}
where $\tilde{z}_{k,q} = \mathbf{x}^H \left(\mathbf{D}_{M,q}^\ast\otimes\mathbf{J}_{N,k}^T\right)\mathbf{z}$. One may then define the FST-AF as
\begin{equation}
    \mathcal{A}_{\rm{FST}}\left(k,q\right) :=  \mathbf{x}^H\left(\mathbf{D}_{M,q}^\ast\otimes\mathbf{J}_{N,k}^T\right)\mathbf{x},\quad k\in\mathbb{Z}_N,q\in\mathbb{Z}_M,
\end{equation}
using which the MF output can be expressed again as the linear combination of $L$ time-frequency-shifted FST-AFs.

To proceed with a more compact formulation, let us represent the frequency-slow-time domain signal of the FST representation as
\begin{equation}
    \mathbf{X}_{\text{FT}}: = \mathbf{F}_N\mathbf{X} = \mathbf{F}_N\mathbf{U}\mathbf{S}\mathbf{V}.
\end{equation}
Given the fact that the periodic time shift matrix can be diagonalized by the DFT matrix, i.e.,
\begin{equation}
    \mathbf{J}_{N,k} = \mathbf{F}_N^H\mathbf{D}_{N,k}^\ast\mathbf{F}_N,
\end{equation}
the FST-AF can be further written as
\begin{align}
    \nonumber\mathcal{A}_{\rm{FST}}\left(k,q\right) &= \text{vec}\left(\mathbf{X}\right)^H\left(\mathbf{D}_{M,q}^\ast\otimes\mathbf{J}_{N,k}^T\right)\text{vec}\left(\mathbf{X}\right)\\
    & \nonumber = \text{vec}\left(\mathbf{X}\right)^H\text{vec}\left(\mathbf{J}_{N,k}^T\mathbf{X}\mathbf{D}_{M,q}^\ast\right)\\
    & \nonumber = \text{vec}\left(\mathbf{F}_N^H\mathbf{X}_{\text{FT}}\right)^H\text{vec}\left(\mathbf{F}_N^H\mathbf{D}_{N,k}\mathbf{X}_{\text{FT}}\mathbf{D}_{M,q}^\ast\right)\\
    & \nonumber =\text{Tr}\left(\mathbf{D}_{N,k}\mathbf{X}_{\text{FT}}\mathbf{D}_{M,q}^\ast\mathbf{X}_{\text{FT}}^H\right)\\
    & 
    = \sqrt{MN}\mathbf{f}_{N,k+1}^H\left|\mathbf{X}_{\text{FT}}\right|^2\mathbf{f}_{M,q+1}. \label{FSTAF_simplified}
\end{align}
By noting \eqref{FSTAF_simplified}, the FST-AF may be naturally represented in a matrix form as 
\begin{equation}\label{A_FST}
    \mathbf{A}_{\rm{FST}}: = \mathbf{F}_N^H\left|\mathbf{X}_{\text{FT}}\right|^2\mathbf{F}_M,
\end{equation}
with its $(n,m)$-th element being
\begin{equation}
    a_{n,m}: = \mathcal{A}_{\rm{FST}}\left(n-1,m-1\right)
\end{equation}

\textbf{Remark 1}. The formulation in \eqref{FSTAF_simplified} admits a clear interpretation. Under the FST representation, where the fast-time Doppler shift is assumed negligible, the delay and Doppler contributions become decoupled and act only along the fast-time and slow-time axes, respectively. This corresponds to a small-Doppler approximation of the DP-AF. In this setting, the delay-Doppler plane serves as the Fourier-dual domain of the frequency-slow-time representation in the 2D-DFT sense. Since the AF can be interpreted as the two-dimensional circular self-convolution of the delay-Doppler-domain signal, it follows that the AF is equivalently obtained by taking the 2D-DFT of the squared magnitude of the FT-domain signal.

Again, due to the random nature of the FST-AF, we are interested in characterizing its EISL and ESL. The expectation of the squared Frobenius norm of $\mathbf{A}$ can be calculated as
\begin{align}\label{FST_AF_volume}
     &\nonumber \mathbb{E}\left(\left\|\mathbf{A}_{\rm{FST}}\right\|_F^2\right) = MN\mathbb{E}\left(\left\|\mathbf{F}_N^H\left|\mathbf{F}_N\mathbf{U}\mathbf{S}\mathbf{V}\right|^2\mathbf{F}_M\right\|_F^2\right)\\
    & = MN\mathbb{E}\left(\left\|\left|\mathbf{F}_N\mathbf{U}\mathbf{S}\mathbf{V}\right|^2\right\|_F^2\right) = MN\mathbb{E}\left(\left\|\mathbf{F}_N\mathbf{U}\mathbf{S}\mathbf{V}\right\|_4^4\right),    
\end{align}
where the expectation is over the random symbol matrix $\mathbf{S}$. 
Accordingly, the ESL of the FST-AF is 
\begin{align}
    &\nonumber \mathbb{E}\left(|\mathcal{A}_{\rm{FST}}\left(k,q\right)|^2\right)=\mathbb{E}(|a_{k+1,q+1}|^2) 
    \\&= MN \mathbb{E}\left(\left|\mathbf{f}_{N,k+1}^H\left|\mathbf{F}_N\mathbf{U}\mathbf{S}\mathbf{V}\right|^2\mathbf{f}_{M,q+1}\right|^2\right), \quad (k,q)\ne (0,0).
\end{align}
The EISL is therefore given by
\begin{align}\label{EISL_FST_definition}
   \nonumber {\rm EISL}_{\rm{FST}}  &=\sum_{k=0}^{N-1} \sum_{q=0}^{N-1} \mathbb{E}(|\mathcal{A}_{\rm{FST}}(k,q)|^2) - \mathbb{E}(|\mathcal{A}_{\rm{FST}}(0,0)|^2)\\
    &=\mathbb{E}\left(\left\|\mathbf{A}_{\rm{FST}}\right\|_F^2\right) - \mathbb{E}(|a_{1,1}|^2),
\end{align}
where the expected mainlobe level is
\begin{equation}\label{FST_mainlobe}
    \mathbb{E}(|\mathcal{A}_{\rm{FST}}(0,0)|^2) = \mathbb{E}(|a_{1,1}|^2) = M^2N^2 + (\kappa-1)MN.
\end{equation}

\section{Statistical Characterization of the DP-AF for Random Communication Signals} \label{sec_3}

In this section, we derive the closed-form expressions of the EISL and ESL of the DP-AF for random communication signals, and further establish several fundamental properties of the DP-AF.

\subsection{Main Results}
Let us first investigate the invariance of the EISL for DP-AF.
\begin{proposition}[Invariance of EISL] For all constellations and modulation schemes, the normalized $\rm EISL_{\rm{DP}}$ is a constant value $N-1$. \label{invariance_EISL_prop}
\begin{equation}
    \frac{\sum_{k=0}^{N-1} \sum_{q=0}^{N-1}\mathbb{E}(\left| \mathcal{A}_{\rm{DP}}(k,q) \right|^2) - \mathbb{E}(\left| \mathcal{A}_{\rm{DP}}(0,0) \right|^2)}{\mathbb{E}(\left| \mathcal{A}_{\rm{DP}}(0,0) \right|^2)} = N -1.
\end{equation}
\end{proposition} 
\begin{proof}
According to \eqref{2D AF}, the volume of the DP-AF is
\begin{align}\label{volume of the DP-AF}
    &\nonumber \sum_{k=0}^{N-1} \sum_{q=0}^{N-1} \left| \mathcal{A}_{\rm{DP}}(k,q) \right|^2  \\
    &\nonumber = N  \sum_{k=0}^{N-1} \sum_{q=0}^{N-1} \mathbf{x}^H \mathbf{J}_{N,k}^T \mathrm{Diag}(\mathbf{x}) \mathbf{f}_{N,q+1} \mathbf{f}_{N,q+1}^H \mathrm{Diag}(\mathbf{x}^\ast) \mathbf{J}_{N,k} \mathbf{x} \\
    &\nonumber = N  \sum_{k=0}^{N-1} (\mathbf{J}_{N,k} \mathbf{x})^H \mathrm{Diag}(\mathbf{x}) \mathrm{Diag}(\mathbf{x}^\ast) \mathbf{J}_{N,k} \mathbf{x} \\
    &\nonumber = N \mathbf{x}^T \left[ \sum_{k=0}^{N-1} \mathrm{Diag}(\mathbf{J}_{N,k} \mathbf{x})^H \mathrm{Diag}(\mathbf{J}_{N,k} \mathbf{x}) \right] \mathbf{x}^\ast\\
    &= N \| \mathbf{x} \|_2^4 
    = N \left| \mathcal{A}_{\rm{DP}}(0,0) \right|^2.
\end{align}
It is observed that the volume of DP-AF is $N$ times that of its mainlobe, which directly follows from the invariance of the volume of Woodward's AF \cite{woodward1953probability}. Then we have
\begin{align}
    \mathrm{EISL}_{\rm{DP}} &\nonumber=(N-1) \mathbb{E}(|\mathcal{A}_{\rm{DP}}(0,0)|^2)\\
    & = (N^2-N) \left(N+\kappa - 1\right).
\end{align}
Therefore, $\rm EISL_{\rm{DP}}$ remains invariant regardless of the modulation waveforms employed.
\end{proof}
According to Propositon \ref{invariance_EISL_prop}, communication-centric ISAC signals constructed by arbitrary constellation and modulation basis yield the same normalized $\rm{EISL}_{\rm{DP}}$. In such a sense, no optimal waveform exists in terms of minimizing the summation of all the ESL. 

We next derive a closed-form expression for the average squared DP-AF.
\begin{proposition}\label{prop_ESL_DP_AF}
The closed-form expression of the average squared DP-AF is
    \begin{align}\label{the average sidelobe}
        &\nonumber \mathbb{E}(\left| \mathcal{A}_{\rm{DP}}(k,q) \right|^2) \\
        & = N + (\kappa - 2)  \sum_{n=1}^N \left| \mathbf{u}_n^H \mathbf{D}_{N,q}\mathbf{J}_{N,k}\mathbf{u}_n \right|^2 + N^2\delta_{k,0} \delta_{q,0}.
    \end{align}
where $(k,q)\in\mathbb{Z}_N^2$.

\begin{proof}
    See Appendix \ref{proof of squared AF}
\end{proof}
\end{proposition}

According to the Cauchy-Schwarz inequality, we have 
\begin{align} 
&\nonumber\left| \mathbf{u}_n^H \mathbf{D}_{N,q} \mathbf{J}_{N,k} \mathbf{u}_n \right|^2 \leq \left\| \mathbf{u}_n\right\|^2 \left\| \mathbf{D}_{N,q} \mathbf{J}_{N,k} \mathbf{u}_n\right\|^2 = 1.
\end{align}
This implies
\begin{equation}
    0\leq\sum_{n=1}^N \left| \mathbf{u}_n^H \mathbf{D}_{N,q}\mathbf{J}_{N,k} \mathbf{u}_n \right|^2\leq N,
\end{equation}
which indicates that for sub-Gaussian constellations ($\kappa < 2$), the ESL of the DP-AF is bounded by
\begin{equation}\label{ESL_LB}
    \left(\kappa-1\right)N\leq\mathbb{E}(\left| \mathcal{A}_{\rm{DP}}(k,q) \right|^2)\leq N ,\quad \forall (k,q) \ne (0,0).
\end{equation}
Since the kurtosis can be made arbitrarily close to $1$ through constellation shaping \cite{10685511}, the lower bound can in principle be made arbitrarily close to zero. Moreover, note that
\begin{align}
    \left\|\mathbf{U}^H\mathbf{D}_{N,q} \mathbf{J}_{N,k}\mathbf{U}\right\|_F^2 = N,
\end{align}
and that $\mathbf{u}_n^H \mathbf{D}_{N,q}\mathbf{J}_{N,k} \mathbf{u}_n$ corresponds to the $n$th diagonal element of $\mathbf{U}^H\mathbf{D}_{N,q} \mathbf{J}_{N,k}\mathbf{U}$. It then follows that the lower bound in \eqref{ESL_LB} is achieved if and only if $\mathbf{D}_{N,q} \mathbf{J}_{N,k}$ is diagonalizable by $\mathbf{U}$. Two trivial examples here are the matrices set $\left\{\mathbf{D}_{N,q} \mathbf{J}_{N,0}\right\}_{q \in\mathbb{Z}_N}$ and $\left\{\mathbf{D}_{N,0} \mathbf{J}_{N,k}\right\}_{k\in\mathbb{Z}_N}$, which are respectively diagonalized by $\mathbf{U} =  \mathbf{I}_N$ and $\mathbf{U} =  \mathbf{F}_N^H$. This implies that the SC and OFDM waveforms achieve the minimum ESL at the zero-Doppler and zero-delay cuts of the DP-AF, respectively.

We now examine how these lowest sidelobes are distributed over the delay-Doppler domain under different waveform designs. Since no orthogonal waveform can minimize the overall EISL of the DP-AF, a more meaningful question is whether one can design a waveform that minimizes the sidelobe level within a prescribed delay-Doppler region. As illustrated by the SC and OFDM examples above, such minimization is indeed achievable when the region of interest is constrained to a one-dimensional subset of the delay-Doppler plane, such as the zero-delay or zero-Doppler cuts. In practical settings, however, it is often more desirable to concentrate the lowest sidelobes over a genuinely 2D region formed by contiguous delay-Doppler bins. This motivates the concept of a low-ambiguity zone \cite{10935731}, namely a designated connected 2D region in the delay-Doppler domain where the waveform is designed to exhibit minimized ESL. Waveforms constructed in this manner may thus be regarded as ``locally'' optimal with respect to the specified 2D region. As we shall see, however, even such locally optimal waveforms do not generally exist, because the geometry of the lowest sidelobes is intrinsically constrained by the algebraic structure of the WH group. This is formalized in the following theorem.
\begin{thm}\label{thm_no_LAZ}
Given an orthogonal modulation basis $\mathbf{U}\in\mathcal{U}(N)$, define
    \begin{equation}
        \mathcal{G}_{\mathbf{U}}:=\left\{e^{j\ell\frac{2\pi}{N}}\mathbf{G}_{k,q}|\mathbf{U}^H\mathbf{G}_{k,q}\mathbf{U}\;\text{is diagonal,}\;\;(\ell,k,q)\in\mathbb{Z}_N^3\right\},
    \end{equation}
    as well as the index set
    \begin{equation}
        \mathcal{I}_{\mathbf{U}}:= \left\{\left(k,q\right)\in\mathbb{Z}_N^2|\mathbf{G}_{k,q}\in\mathcal{G}_{\mathbf{U}}\right\}. 
    \end{equation}
    For any 
    $(a,b) \in \mathbb{Z}_N^2$, it is impossible that
    \[
    (a,b),\;\;(a,\langle b+1 \rangle_N),\;\;(\langle a+1 \rangle_N,b)
    \]
    all belong to $\mathcal{I}_\mathbf{U}$ simultaneously.
\end{thm}
\begin{proof}
    See Appendix \ref{appendix_thm_1}.
\end{proof}
Theorem 1 suggests that no orthogonal waveform can simultaneously minimize the sidelobe levels over a $2\times 2$ region in the delay-Doppler domain. Consequently, it is impossible for any waveform to generate a valid low-ambiguity zone pattern that achieves the minimal ESL in any specified 2D compact region over the delay-Doppler plane.

We further show that, for any waveform $\mathbf{U}$, the number of lowest sidelobes in the DP-AF does not exceed $N-1$, which follows from the proposition below.
\begin{proposition}\label{prop_size_S}
    Given an orthogonal modulation basis $\mathbf{U}\in\mathcal{U}(N)$, it holds that
    \begin{equation}
        \mathrm{Card}\left(\mathcal{I}_{\mathbf{U}}\right) \leq N.
    \end{equation}
\end{proposition}
\begin{proof}
    See Appendix \ref{proof_prop_size_S}.
\end{proof}
Proposition \ref{prop_size_S} indicates that at most $N$ delay-Doppler matrices are diagonalizable by $\mathbf{U}$. By noting the fact that $\mathbf{G}_{0,0} = \mathbf{I}_N\in\mathcal{G}_{\mathbf{U}}$ contributes solely to the mainlobe, i.e., $(0,0)$, the remaining elements in $\mathcal{I}_{\mathbf{U}}$ correspond to at most $N-1$ lowest sidelobes.

Next, we use the following proposition to characterize the distribution of these minimum sidelobes over the delay-Doppler plane.
\begin{proposition}\label{dispersion_properties}
    Any three points in the index set $\mathcal{I}_{\mathbf{U}}$ are either collinear modulo $N$, or form a triangle with an area no less than $N/2$.
\end{proposition}
\begin{proof}
    See Appendix \ref{dispersion_proof}.
\end{proof}
\textbf{Remark 2}. In a nutshell, Theorem \ref{thm_no_LAZ}, Proposition \ref{prop_size_S}, and Proposition \ref{dispersion_properties} reveal a fundamental limitation in shaping the DP-AF sidelobe pattern. Theorem \ref{thm_no_LAZ} rules out the possibility of an orthogonal waveform that attains the minimum ESL over any $2\times 2$ delay-Doppler block, thereby excluding the existence of a concentrated low-ambiguity zone. Proposition \ref{prop_size_S} further shows that, for any modulation basis $\mathbf{U}$, there are at most $N-1$ lowest sidelobes in the entire delay-Doppler plane. Proposition \ref{dispersion_properties} then constrains their geometry, implying that these minimum sidelobes must either lie on a single mod-$N$ line or be highly dispersed in the sense that any three non-collinear points enclose an area of at least $N/2$. These results jointly indicate that no waveform can globally minimize the ESL of the DP-AF, nor can one concentrate the minimum sidelobes over a compact 2D region. At best, the minimum ESL can be attained along a one-dimensional cut (such as zero-delay or zero-Doppler cuts) or within a set of widely separated delay-Doppler bins.

\subsection{Case Study}
In this subsection, we derive closed-form expressions for the average squared DP-AF corresponding to four representative communication waveforms. These results provide further quantitative evidence supporting the conclusions discussed above.
\subsubsection{Average Squared DP-AF for OFDM}
\begin{corollary}
    The average squared DP-AF of OFDM is
\begin{align}\label{ofdm_dp_af}
    \mathbb{E}\left(\left| \mathcal{A}_{\rm{DP}}^{\rm{OFDM}}(k,q) \right|^2\right) = N + (\kappa - 2)N  \delta_{q,0} + N^2 \delta_{k,0} \delta_{q,0}.
\end{align}
\end{corollary}
\begin{proof}
    By letting $\mathbf{U} = \mathbf{F}_N^H$, it is straightforward to see that
    \begin{align}\label{ofdm_waveform_term}
        &\nonumber\sum_{n=1}^N \left| \mathbf{u}_n^H \mathbf{D}_{N,q}\mathbf{J}_{N,k}\mathbf{u}_n \right|^2 = \operatorname{Tr}\left(\left|\mathbf{U}^H\mathbf{D}_{N,q}\mathbf{J}_{N,k}\mathbf{U}\right|^2\right)\\
        &\nonumber = \operatorname{Tr}\left(\left|\mathbf{F}_N\mathbf{D}_{N,q}\mathbf{F}_N^H\mathbf{D}_{N,k}^\ast\mathbf{F}_N\mathbf{F}_N^H\right|^2\right)\\
        &\nonumber = \operatorname{Tr}\left(\left|\mathbf{F}_N\mathbf{D}_{N,q}\mathbf{F}_N^H\mathbf{D}_{N,k}^\ast\right|^2\right) \\&= \operatorname{Tr}\left(\left|\mathbf{J}_{N,q}\mathbf{D}_{N,k}^\ast\right|^2\right) = N\delta_{q,0}. 
    \end{align}
Substituting \eqref{ofdm_waveform_term} into \eqref{the average sidelobe} yields \eqref{ofdm_dp_af}.
\end{proof}
The average squared zero-Doppler slice of OFDM is
\begin{align}
    \mathbb{E}\left(\left| \mathcal{A}_{\rm{DP}}^{\rm{OFDM}}(k,0) \right|^2\right)= (\kappa - 1)N.
\end{align}
The average squared zero-delay slice of OFDM is
\begin{equation}
\begin{aligned}
    \mathbb{E}\left(\left| \mathcal{A}_{\rm{DP}}^{\rm{OFDM}}(0,q) \right|^2\right)= N.
\end{aligned}
\end{equation}
Moreover, the ESL of OFDM, excluding the zero-Doppler and zero-delay slices ($k\ne 0, q\ne 0$), is given by
\begin{align}
    \mathbb{E}\left(\left| \mathcal{A}_{\rm{DP}}^{\rm{OFDM}}(k,q) \right|^2\right) = N.
\end{align}
As indicated in the previous discussion, for sub-Gaussian constellations ($1 \leq \kappa <2$), the $N-1$ lowest sidelobes are all placed in the zero-Doppler slice of the DP-AF of OFDM. This is consistent with the results in \cite{liu2024ofdmachieveslowestranging}, that CP-OFDM achieves the lowest average ranging sidelobe among all orthogonal waveforms with CP. 

\subsubsection{Average Squared DP-AF for SC}
\begin{corollary}
    The average squared DP-AF of SC is
\begin{align}\label{SC_DP_AF}
    \mathbb{E}\left(\left| \mathcal{A}_{\rm{DP}}^{\rm{SC}}(k,q) \right|^2\right) = N + (\kappa - 2)N \delta_{k,0} + N^2 \delta_{k,0} \delta_{q,0}.
\end{align}
\end{corollary}
\begin{proof} 
Let $\mathbf{U} = \mathbf{I}_N$. It follows that
\begin{equation}
    \nonumber\sum_{n=1}^N \left| \mathbf{u}_n^H \mathbf{D}_{N,q}\mathbf{J}_{N,k}\mathbf{u}_n \right|^2 = \operatorname{Tr}\left(\left|\mathbf{D}_{N,q}\mathbf{J}_{N,k}\right|^2\right) = N\delta_{k,0},
\end{equation}
which results in \eqref{SC_DP_AF}.
\end{proof}
The average squared zero-Doppler slice of SC is
\begin{equation}
    \mathbb{E}\left(\left| \mathcal{A}_{\rm{DP}}^{\rm{SC}}(k,0) \right|^2\right) = N.
\end{equation}
The average squared zero-delay slice of SC is
\begin{equation}
    \mathbb{E}\left(\left| \mathcal{A}_{\rm{DP}}^{\rm{SC}}(0,q) \right|^2\right) =
    (\kappa - 1)N.
\end{equation}
The ESL of SC, excluding the zero-Doppler and zero-delay slices ($k\ne 0, q\ne 0$), is
\begin{equation}
    \mathbb{E}\left(\left| \mathcal{A}_{\rm{DP}}^{\rm{SC}}(k,q) \right|^2\right) = N.
\end{equation}
In contrast to the OFDM modulation, for sub-Gaussian constellations, SC achieves the lowest Doppler sidelobe levels.

\subsubsection{Average Squared DP-AF for OTFS}
\begin{proposition}\label{the average squared DP-AF of OTFS}
    The average squared DP-AF of OTFS is
    \begin{align}\label{DP_AF_OTFS}
    \nonumber& \mathbb{E}\left(\left| \mathcal{A}_{\rm{DP}}^{\rm{OTFS}}(k,q) \right|^2\right) \\
        &= N + (\kappa - 2)N \cdot \delta_{\langle 
     k\rangle_{N_2}, 0} \cdot \delta_{\langle q \rangle_{N_1}, 0} + N^2 \delta_{k,0} \delta_{q,0}.
    \end{align}
\end{proposition}
\begin{proof}
    See Appendix \ref{the proof of OTFS}.
\end{proof}

The average squared zero-Doppler slice of OTFS is
\begin{equation}
        \mathbb{E}\left(\left| \mathcal{A}_{\rm{DP}}^{\rm{OTFS}}(k,0) \right|^2\right) = N + (\kappa - 2)N  \cdot\delta_{\langle k\rangle_{N_2}, 0}.
\end{equation}
The average squared zero-delay slice of OTFS is
\begin{equation}
    \mathbb{E}\left(\left| \mathcal{A}_{\rm{DP}}^{\rm{OTFS}}(0,q) \right|^2\right) = N + (\kappa - 2)  N \cdot \delta_{\langle q\rangle_{N_1}, 0}.
\end{equation}
Excluding the zero-Doppler and zero-delay slices, the ESL ($k \ne 0, q \ne 0$) of OTFS is
\begin{equation}
    \mathbb{E}\left(\left| \mathcal{A}_{\rm{DP}}^{\rm{OTFS}}(k,q) \right|^2\right)  = N + (\kappa - 2)N \cdot \delta_{\langle k\rangle_{N_2}, 0} \cdot \delta_{ \langle q \rangle_{N_1},0}.
\end{equation}

It is worth noting that the $N-1$ lowest sidelobes together with the mainlobe are generally placed over an $N_1\times N_2$ uniform grid in the DP-AF of the OTFS waveform. For any set of four nearest sidelobes, their delay-Doppler indices are either collinear or form a rectangle whose area equals $N$. Consequently, any three nearest sidelobes that are not collinear form a triangle of area $N/2$, which is consistent with the result in Proposition \ref{prop_size_S}.

\subsubsection{Average Squared DP-AF for AFDM}
\begin{proposition}\label{afdm_sp_af_squared}
    The average squared DP-AF of AFDM is
    \begin{align}\label{dp_af_afdm}
        &\nonumber \mathbb{E}\left(\left|\mathcal{A}_{\rm{DP}}^{\rm{AFDM}}(k,q)\right|^2\right)\\
        &= N + (\kappa -2)N \cdot \delta_{\langle 2Nkc_1 -q\rangle_N, 0} + N^2 \delta_{k,0} \delta_{q,0}.
    \end{align}    
\end{proposition}
\begin{proof}
    See Appendix \ref{the proof of AFDM}.
\end{proof}

Accordingly, the average squared zero-Doppler slice of AFDM is
\begin{equation}
        \mathbb{E}\left(\left| \mathcal{A}_{\rm{DP}}^{\rm{AFDM}}(k,0) \right|^2\right) = N + (\kappa - 2) N \cdot \delta_{\langle 2Nk{c_1} \rangle_N, 0}. 
\end{equation}
The average squared zero-delay slice of AFDM is
\begin{equation}
    \mathbb{E}\left(\left| \mathcal{A}_{\rm{DP}}^{\rm{AFDM}}(0,q) \right|^2\right) = N.
\end{equation}
The ESL of the remaining sidelobes ($k \ne 0, q \ne 0$) of AFDM is
\begin{equation}
        \mathbb{E}\left(\left| \mathcal{A}_{\rm{DP}}^{\rm{AFDM}}(k,q) \right|^2\right)  = N + (\kappa -2)N \cdot  \delta_{\langle 2Nkc_1 - q\rangle_N, 0}.
\end{equation}

It follows from \eqref{dp_af_afdm} that, for AFDM, the locations of the lowest sidelobes are determined by
\begin{equation}\label{modN_line}
    \big\langle q - 2N c_{1} \cdot k \big\rangle_{N} = 0, \quad (q,k) \in \mathbb{Z}_{N}^{2},
\end{equation}
which represents a line modulo $N$ in the delay-Doppler domain with slope $2Nc_{1} \in \mathbb{Z}$. Since $k$ traverses all values in $\mathbb{Z}_N$, the line \eqref{modN_line} contains exactly $N$ lattice points in $\mathbb{Z}_{N}^{2}$, one of which corresponds to the mainlobe at $(0,0)$. Hence, the AFDM waveform exhibits $N - 1$ lowest sidelobes in the average squared DP-AF. Moreover, by basic planar geometry, any three nearest non-collinear lattice points on this line form a triangle with area $N/2$.

\begin{figure}[!t]
    \label{4_waveforms_af}
    \centering
    \vspace{-0.3cm}
        \subfigure[OFDM]{
        \includegraphics[width = 0.45\columnwidth]{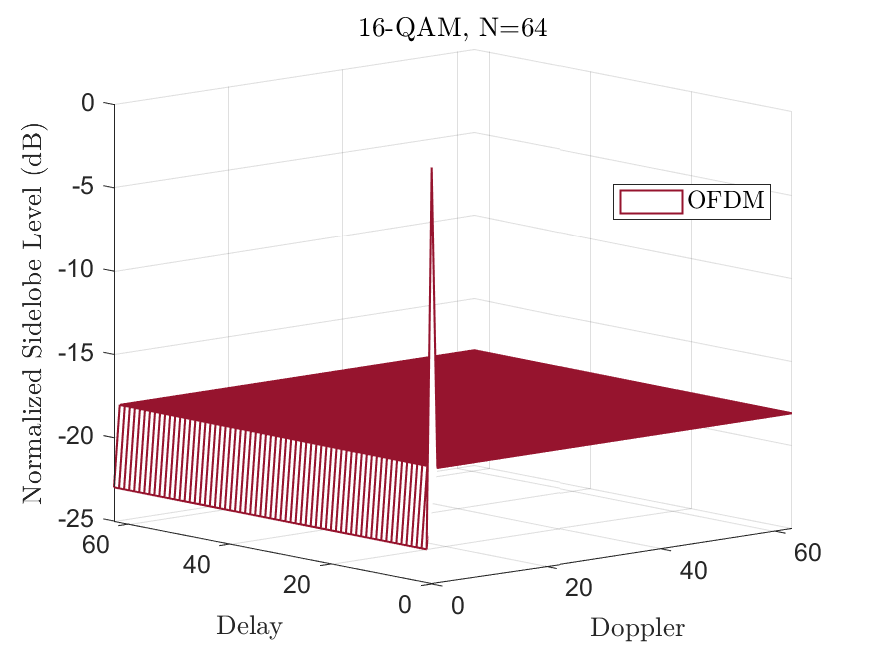} 
        \label{fig: AF of 16QAM under OFDM compare}    
    }
    \subfigure[SC]{
        \includegraphics[width=0.45\columnwidth]{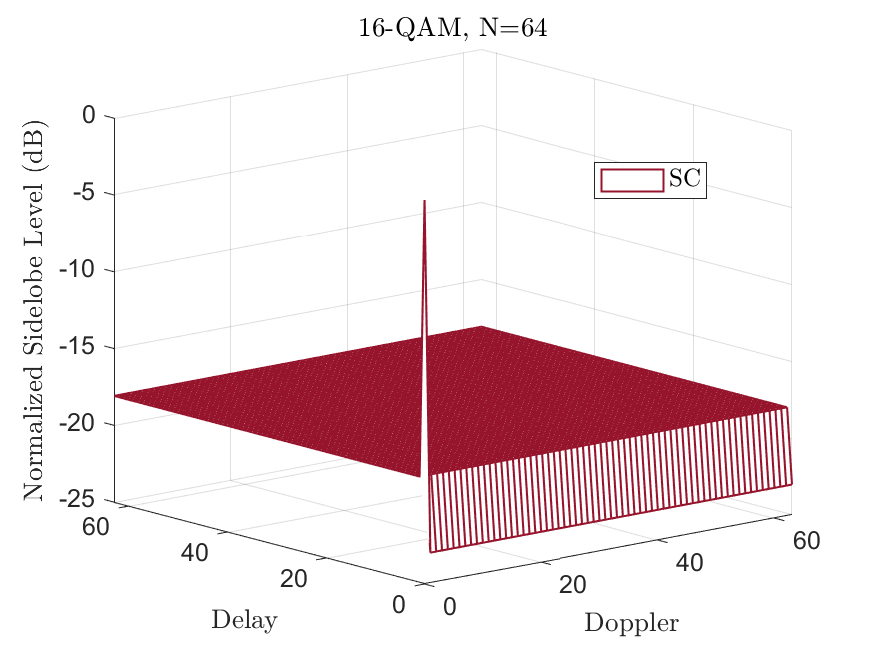}
        \label{fig: AF of 16QAM under SC compare}    
    }

    \subfigure[OTFS]{
        \includegraphics[width = 0.45\columnwidth]{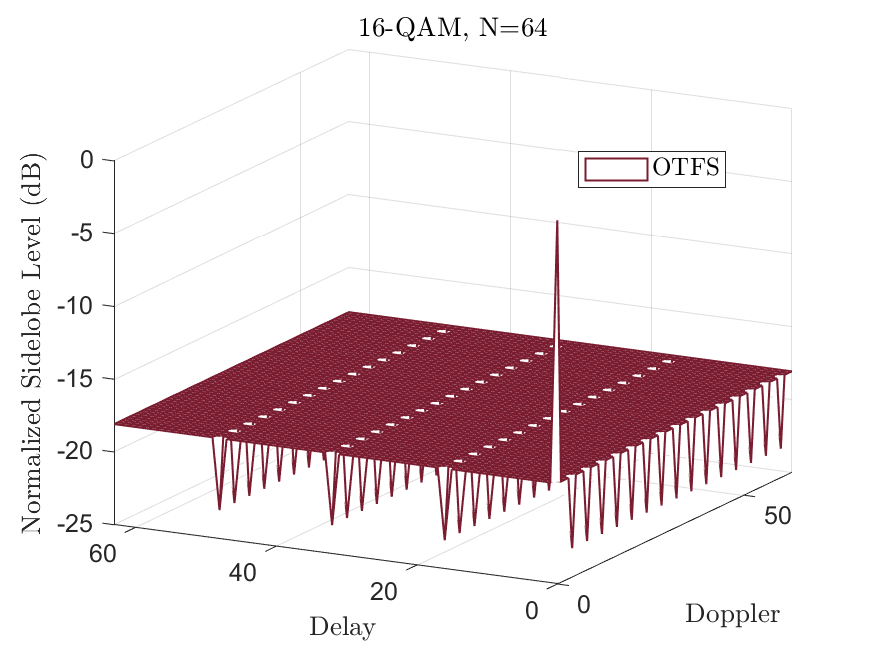} 
        \label{fig: AF of 16QAM under OTFS compare}    
        }
    \subfigure[AFDM]{
        \includegraphics[width=0.45\columnwidth]{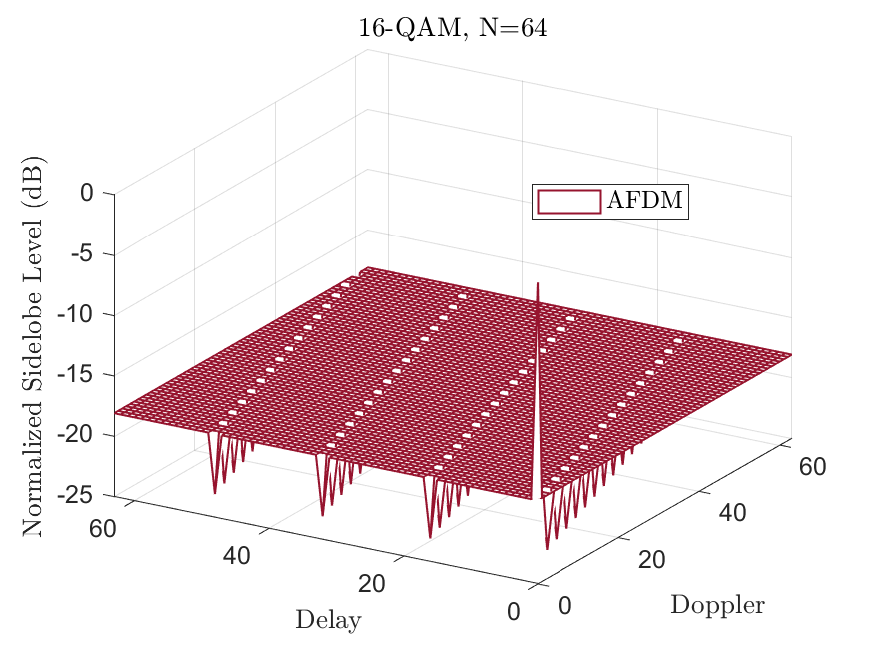}
        \label{fig: AF of 16QAM under AFDM compare}    
        }
    \caption{The average squared DP-AF of OFDM, SC, OTFS and AFDM waveforms with $N = 64$ under 16-QAM constellation.}
    \label{fig: the 2D AF performances compare of OTFS and AFDM}
\end{figure}
\textbf{Remark 3}. Fig. \ref{fig: the 2D AF performances compare of OTFS and AFDM} illustrates the average squared DP-AFs of the four waveforms under comparison, adopting a standard 16-QAM constellation with kurtosis $\kappa = 1.32$. The ESL values are normalized with respect to the average mainlobe level. The number of random symbols is set to $N = 64$. For the OTFS waveform, we adopt $N_1 = 4, N_2 = 16$, and for the AFDM waveform, we use $c_1 = \frac{1}{32}, c_2 = \frac{1}{16}$. As predicted by the theoretical results, the ESL of all DP-AFs are binary valued: The low sidelobes correspond exactly to those delay-Doppler matrices that are diagonalizable by the respective modulation basis. Moreover, each DP-AF contains exactly $63$ such low sidelobes. These sidelobes are either aligned along a line, or, if not collinear, any set of three sidelobes encloses a triangular region with area at least $32$.

\textbf{Remark 4}. It is worth noting that the above discussion on DP-AFs primarily concerns waveforms modulated with sub-Gaussian constellations, which are more commonly used in modern wireless systems. For completeness, we mention that the normalized EISL remains invariant under super-Gaussian constellations as well. However, in contrast to the sub-Gaussian case, the delay-Doppler matrices that are diagonalizable by the modulation basis now yield the {\textit{largest}} ESL. This follows directly from the bounds in \eqref{ESL_LB}: when $\kappa > 2$, the lower and upper bounds reverse. As a consequence, one still cannot construct a compact low-ambiguity region for waveforms modulated carrying super-Gaussian constellation symbols.

\section{Statistical Characterization of the FST-AF for Random Communication Signals}\label{sec_4}
In this section, we characterize the EISL and ESL of the FST-AF for random communication signals, and then establish the optimality of OFDM and OTFS under sub- and super-Gaussian constellations, respectively.
\subsection{Main Results}
Let us first express ${\rm EISL}_{\rm{FST}}$ and the average squared FST-AF in closed forms.

\begin{proposition}\label{prop_EISL_FST_AF}
    The EISL for FST-AF is given by
    \begin{align}\label{the EISL of stop-and-go model}
    {\rm EISL}_{\rm{FST}} &\nonumber=
     M^2N^2-MN \\
     &+ \left(\kappa - 2\right)MN\left(\left\|\mathbf{V}^T\otimes\mathbf{F}_N\mathbf{U}\right\|_4^4 - 1\right).      
    \end{align}
\end{proposition}
\begin{proof}
    See Appendix \ref{proof of the EISL of FST-AF}.
\end{proof}

\begin{proposition}\label{prop_ESL_FST_AF}
    The average squared FST-AF is
    \begin{align}\label{the average sidelobe of stop-and-go model}
    &\nonumber \mathbb{E}\left( \left| \mathcal{A}_{\rm{FST}} (k,q) \right|^2 \right) 
    =  MN + M^2N^2\delta_{k,0}\delta_{q,0} \\ & + \left(\kappa - 2\right)MN\left\| \left| \mathbf{V} \right|^2 \mathbf{f}_{M,q+1}^\ast \right\|^2 \left\| \left|\mathbf{U}^H \mathbf{F}_N^H \right|^2 \mathbf{f}_{N,k+1} \right\|^2,
    \end{align}
\end{proposition}
where $k \in \mathbb{Z}_N$ and $q \in \mathbb{Z}_M$.
\begin{proof}
    See Appendix \ref{proof ot the average sidelobe of RDM}.
\end{proof}

\subsection{Case Study}\label{Case Study of stop-and-go}

We now proceed to analyze the squared FST-AF of four representative communication waveforms, namely, OFDM, SC, OTFS, and AFDM.
\subsubsection{Average Squared FST-AF for OFDM}
\begin{corollary}
    The average squared FST-AF of OFDM and the corresponding EISL are respectively given by
    \begin{align}
        &\mathbb{E}\left( \left| \mathcal{A}_{\rm{FST}}^{\rm{OFDM}} (k,q) \right|^2 \right) =   M^2N^2\delta_{k,0}\delta_{q,0} + \left(\kappa - 1\right)MN,\label{ESL_FST_OFDM}\\
        &{\rm{EISL}}_{\rm{FST}}^{\rm{OFDM}} = MN(MN-1)(\kappa - 1)\label{EISL_FST_OFDM}.
    \end{align}
\end{corollary}

\begin{proof}
For OFDM, we have $\mathbf{U} = \mathbf{F}_N^H$ and $ \mathbf{V} = \mathbf{I}_M$, taking which into \eqref{the EISL of stop-and-go model} directly leads to
\begin{align}
    &\nonumber\rm{EISL}_{\rm{FST}}^{\rm{OFDM}}\\
    &\nonumber  = M^2N^2-MN + \left(\kappa - 2\right)MN\left(\left\|\mathbf{I}_M\otimes\mathbf{F}_N\mathbf{F}_N^H\right\|_4^4 - 1\right)\\ 
    &\nonumber = M^2N^2-MN + \left(\kappa - 2\right)MN\left(MN - 1\right)\\
    & = MN(MN-1)(\kappa - 1).
\end{align}
Moreover, we have
\begin{align}
    &\nonumber \left\| \left| \mathbf{V} \right|^2 \mathbf{f}_{M,q+1}^\ast \right\|^2 \left\| \left|\mathbf{U}^H \mathbf{F}_N^H \right|^2 \mathbf{f}_{N,k+1} \right\|^2 \\
    &= \left\| \mathbf{f}_{M,q+1}^\ast\right\|^2  \left\|\mathbf{f}_{N,k+1}\right\|^2 = 1,
\end{align}
which results in \eqref{ESL_FST_OFDM}.
\end{proof}
In contrast to its DP-AF counterpart, the average squared FST-AF of OFDM exhibits a constant ESL across the entire delay-Doppler domain, i.e., $(\kappa-1)MN$. We will later show that, for sub-Gaussian constellations, OFDM attains the minimum achievable ESL of the FST-AF among all orthogonal waveforms.

\subsubsection{Average Squared FST-AF for SC}

\begin{corollary}
   The average squared FST-AF of SC and the corresponding EISL are respectively given by
    \begin{align}
        &\nonumber\mathbb{E}\left( \left| \mathcal{A}_{\rm{FST}}^{\rm{SC}} (k,q) \right|^2 \right) \\
        &=   MN + 
         M^2N^2\delta_{k,0}\delta_{q,0} + \left(\kappa - 2\right)MN\delta_{k,0},\label{ESL_FST_SC}\\
        &{\rm{EISL}}_{\rm{FST}}^{\rm{SC}} = M^2N^2-MN + \left(\kappa - 2\right)MN\left(M - 1\right).
    \end{align}
\end{corollary}
\begin{proof}
    In such a case, we have $\mathbf{U} = \mathbf{I}_N$ and $\mathbf{V} = \mathbf{I}_M$. Taking $\mathbf{U}$ and $\mathbf{V}$ into \eqref{the EISL of stop-and-go model}, the EISL of SC is
\begin{align}
    &\nonumber \rm{EISL}_{\rm{FST}}^{\rm{SC}} \\
    &=\nonumber M^2N^2-MN + \left(\kappa - 2\right)MN\left(\left\|\mathbf{I}_M\otimes\mathbf{F}_N\mathbf{I}_N\right\|_4^4 - 1\right) \\
    & = M^2N^2-MN + \left(\kappa - 2\right)MN\left(M - 1\right).
\end{align}
Taking $\mathbf{U}$ and $\mathbf{V}$ into \eqref{the average sidelobe of stop-and-go model}, we have
\begin{align}
    &\nonumber \left\| \left| \mathbf{V} \right|^2 \mathbf{f}_{M,q+1}^\ast \right\|^2 \left\| \left|\mathbf{U}^H \mathbf{F}_N^H \right|^2 \mathbf{f}_{N,k+1} \right\|^2 \\
    &= \left\| \mathbf{f}_{M,q+1}^\ast\right\|^2  \left\| \frac{1}{N} \mathbf{1}_N\mathbf{1}_N^T\mathbf{f}_{N,k+1}\right\|^2 = \delta_{k,0},
\end{align}
which leads to \eqref{ESL_FST_SC}.
\end{proof}
The average squared zero-Doppler slice of SC is
    \begin{equation}\label{zero_Doppler_SC_FST}
    \mathbb{E}\left( \left| \mathcal{A}_{\rm{FST}}^{\rm{SC}} (k,0) \right|^2 \right) = MN.
    \end{equation}
The average squared zero-delay slice of SC is
    \begin{equation}\label{zero_delay_SC_FST}
    \mathbb{E}\left( \left| \mathcal{A}_{\rm{FST}}^{\rm{SC}} (0,q) \right|^2 \right) = (\kappa - 1)MN.
    \end{equation}
The ESL of remaining sidelobes of SC can be expressed as
    \begin{equation}\label{remaining_sidelobe_SC_FST}
    \mathbb{E}\left( \left| \mathcal{A}_{\rm{FST}}^{\rm{SC}} (k,q) \right|^2 \right) = MN,\quad (k,q)\ne(0,0).
    \end{equation}

By comparing \eqref{ESL_FST_SC} with its DP-AF counterpart in \eqref{SC_DP_AF}, it can be observed that the two AFs for the SC have exactly the same shape in an average sense. In particular, they both exhibit lower ESL in the zero-delay slice, whereas the ESL values of all other delay-Doppler bins remain the same.

\subsubsection{Average Squared FST-AF for OTFS}

\begin{corollary}
    The average squared FST-AF of OTFS and the corresponding EISL are respectively given by
    \begin{align}
        &\mathbb{E}\left( \left| \mathcal{A}_{\rm{FST}}^{\rm{OTFS}} (k,q) \right|^2 \right) = MN + MN(MN+\kappa-2)\delta_{k,0}\delta_{q,0},\label{FST_AF_OTFS}\\
        &{\rm{EISL}}_{\rm{FST}}^{\rm{OFTS}} = MN(MN-1).
    \end{align}
\end{corollary}
\begin{proof}
    For OTFS modulation we have $\mathbf{U} = \mathbf{I}_N$ and $\mathbf{V} = \mathbf{F}_M^H$, substituting which into \eqref{the EISL of stop-and-go model} yields
\begin{align}
    &\nonumber {\rm{EILS}}_{\rm{FST}}^{\rm{OTFS}} \\
    &=\nonumber M^2N^2-MN + \left(\kappa - 2\right)MN\left(\left\|\mathbf{F}_M^\ast\otimes\mathbf{F}_N\mathbf{I}_N\right\|_4^4 - 1\right)\\
    & = MN(MN-1).    
\end{align}
Furthermore, it can be verified that
\begin{align}
    &\nonumber \left\| \left| \mathbf{V} \right|^2 \mathbf{f}_{M,q+1}^\ast \right\|^2 \left\| \left|\mathbf{U}^H \mathbf{F}_N^H \right|^2 \mathbf{f}_{N,k+1} \right\|^2 \\
    &= \left\| \frac{1}{M} \mathbf{1}_{M}\mathbf{1}_{M}^T \mathbf{f}_{M,q+1}^\ast\right\|^2  \left\| \frac{1}{N} \mathbf{1}_{N}\mathbf{1}_{N}^T \mathbf{f}_{N,k+1}\right\|^2  = \delta_{k,0}\delta_{q,0}.
\end{align}
Plugging the above equation into \eqref{the average sidelobe of stop-and-go model} leads to \eqref{FST_AF_OTFS}.
\end{proof}
An interesting observation is that, for OTFS, the ESL is a constant $MN$ regardless of the employed constellation, since the kurtosis of the constellation contributes solely to the mainlobe.

\subsubsection{Average Squared FST-AF for AFDM}
\begin{proposition}\label{FST_AF_AFDM_prop}
    The average squared FST-AF of AFDM and the corresponding EISL are respectively given by
    \begin{align}
        &\nonumber \mathbb{E}\left( \left| \mathcal{A}_{\rm{FST}}^{\rm{AFDM}} (k,q) \right|^2 \right)  = \\
        &\quad M^2N^2\delta_{k,0}\delta_{q,0}+ MN+ (\kappa -2)MN \cdot \delta_{  \langle 2Nc_1 k \rangle_{N}, 0},\label{FST_AF_AFDM}\\
        & {\rm{EISL}}_{\rm{FST}}^{\rm{AFDM}} = (M\phi-1)(\kappa-1)MN + (N-\phi)M^2N,\label{EISL_FST_AFDM}
    \end{align}
where $\phi = \gcd\left(2Nc_1, N\right)$ is the greatest common divisor between $ 2Nc_1k $ and $N$.
\end{proposition}
\begin{proof}
    See Appendix \ref{proof of AFDM of 'stop-and-go' model}.
\end{proof}
The average squared zero-Doppler slice of SC is
    \begin{equation}\label{zero_Doppler_AFDM_FST}
    \mathbb{E}\left( \left| \mathcal{A}_{\rm{FST}}^{\rm{AFDM}} (k,0) \right|^2 \right) = MN+ (\kappa -2)MN \cdot \delta_{  \langle 2Nc_1 k \rangle_{N}, 0}.
    \end{equation}
The average squared zero-delay slice of AFDM is
    \begin{equation}\label{zero_delay_AFDM_FST}
    \mathbb{E}\left( \left| \mathcal{A}_{\rm{FST}}^{\rm{AFDM}} (0,q) \right|^2 \right) = (\kappa - 1)MN.
    \end{equation}
The ESL of remaining sidelobes of AFDM ($k\ne 0, q\ne 0$) can be expressed as
    \begin{equation}\label{remaining_sidelobe_AFDM_FST}
    \mathbb{E}\left( \left| \mathcal{A}_{\rm{FST}}^{\rm{AFDM}} (k,q) \right|^2 \right) =  MN+ (\kappa -2)MN \cdot \delta_{  \langle 2Nc_1 k \rangle_{N}, 0}.
    \end{equation}

\begin{figure}[!t]
    \label{4_waveforms_fst_af}
    \centering
    \vspace{-0.3cm}
        \subfigure[OFDM]{
        \includegraphics[width = 0.45\columnwidth]{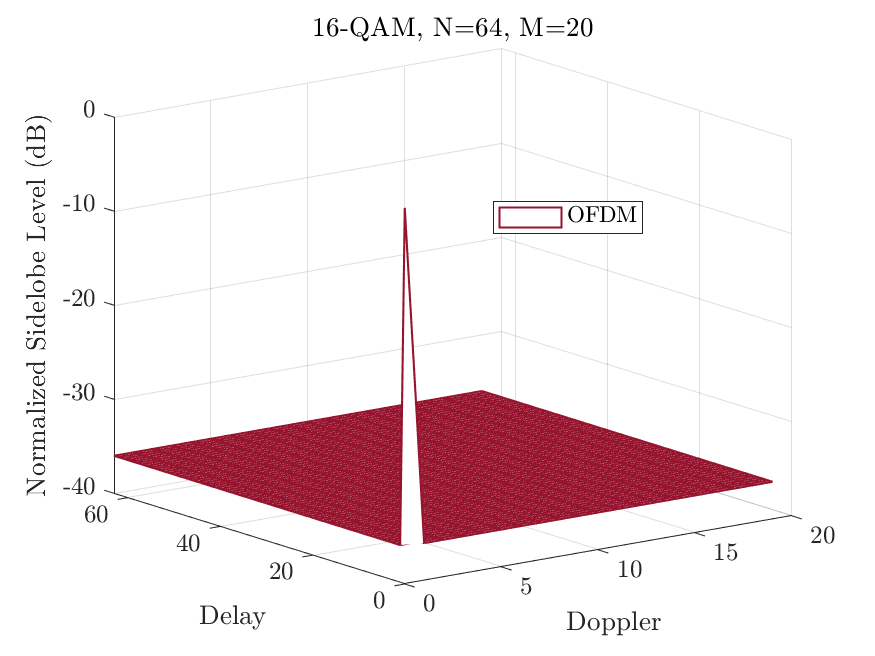} 
        \label{fig: FST-AF of 16QAM under OFDM compare}    
    }
    \subfigure[SC]{
        \includegraphics[width=0.45\columnwidth]{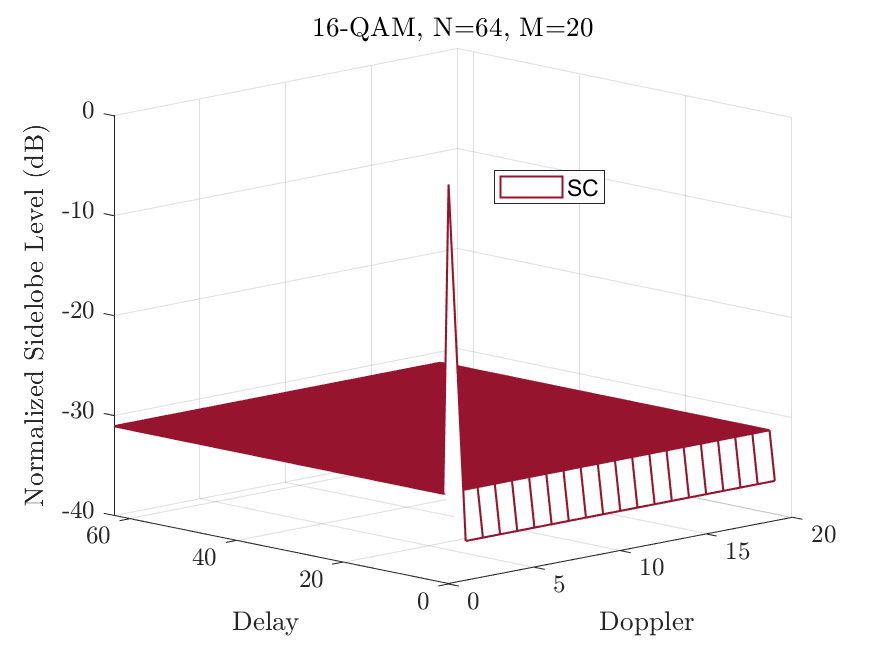}
        \label{fig: FST-AF of 16QAM under SC compare}    
    }

    \subfigure[OTFS]{
        \includegraphics[width = 0.45\columnwidth]{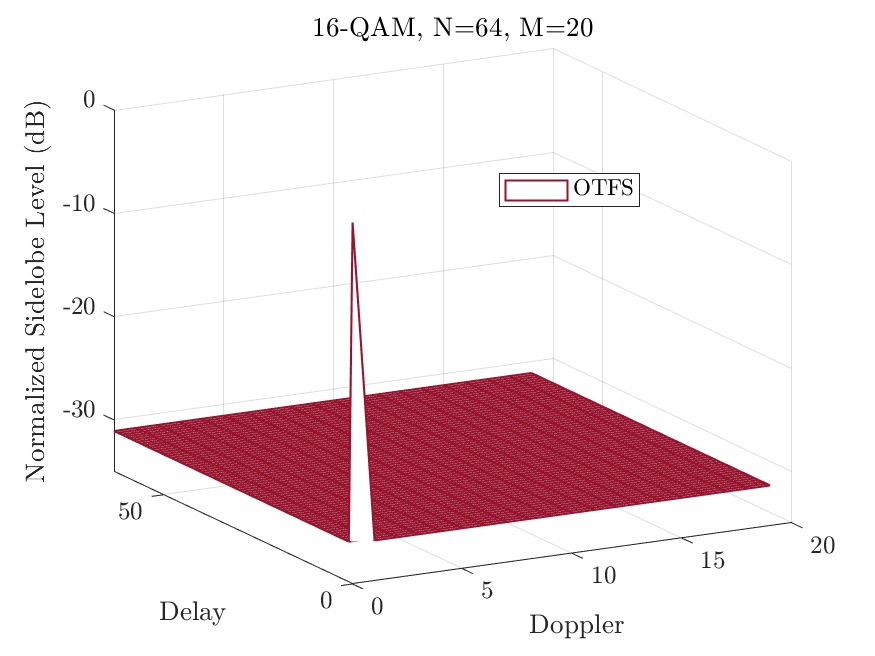} 
        \label{fig: FST-AF of 16QAM under OTFS compare}    
        }
    \subfigure[AFDM]{
        \includegraphics[width=0.45\columnwidth]{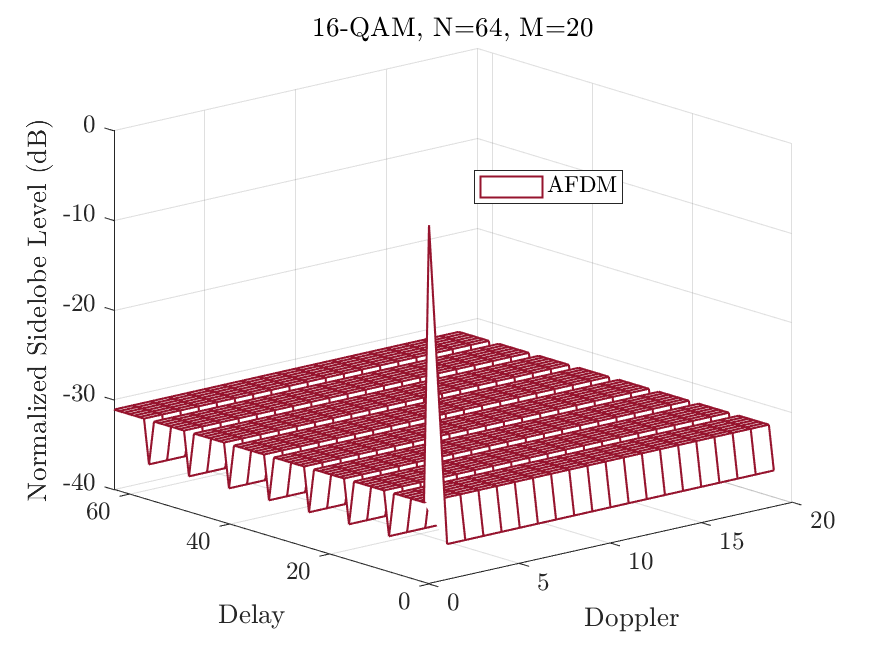}
        \label{fig: FST-AF of 16QAM under AFDM compare}    
        }
    \caption{The average squared FST-AF of OFDM, SC, OTFS and AFDM waveforms with $N = 64$ and $M=20$ under 16-QAM constellation.}
    \label{fig: the FST-AF performances compare}
\end{figure}

\textbf{Remark 5}. Fig. \ref{fig: the FST-AF performances compare} illustrates the resulting average squared FST-AF for the four modulation schemes under 16-QAM constellation with $N = 64$ and $M = 20$. For AFDM, we set $c_1 = \frac{1}{16}$ and $c_2 = \frac{1}{8}$. Unlike the DP-AF case, the FST-AF does not exhibit invariance in its normalized EISL. This indicates that the waveform can, in principle, be optimized to reduce the EISL as well as the ESL globally. In what follows, we establish the optimality of OFDM and OTFS for sub-Gaussian and super-Gaussian constellations, respectively.




\subsection{Optimal Signaling Schemes under FST-AF Regime}
\begin{thm} OFDM and OTFS achieve the lowest ${\rm EISL}_{\rm{FST}}$ under sub-Gaussian and super-Gaussian constellations, respectively.
\end{thm}
\begin{proof}
For sub-Gaussian constellations ($\kappa < 2$), the ${\rm EISL}_{\rm{FST}}$ minimization problem is equivalent to
\begin{equation}
    \max\limits_{\mathbf{U} \in \mathcal{U}(N),\mathbf{V} \in \mathcal{U}(M)}\;\; \left\|\mathbf{V}^T\otimes\mathbf{F}_N\mathbf{U}\right\|_4^4,
\end{equation}
where the maximum is attained if
$\mathbf{V}^T\otimes\mathbf{F}_N\mathbf{U} = \mathbf{I}_{NM}$ \cite{liu2024ofdmachieveslowestranging}. This can be realized by letting
\begin{equation}\label{OFDM_UV}
    \mathbf{V} = \mathbf{I}_M, \quad \mathbf{U} = \mathbf{F}_N^H,
\end{equation}
which corresponds to the OFDM waveform.

For super-Gaussian constellations ($\kappa > 2$), the ${\rm EISL}_{\rm{FST}}$ minimization problem becomes
\begin{equation}
    \min\limits_{\mathbf{U} \in \mathcal{U}(N),\mathbf{V} \in \mathcal{U}(M)}\;\; \left\|\mathbf{V}^T\otimes\mathbf{F}_N\mathbf{U}\right\|_4^4,
\end{equation}
where the minimum is attained if all the entries of $\mathbf{V}^T\otimes\mathbf{F}_N\mathbf{U}$ have constant modulus \cite{liu2024ofdmachieveslowestranging}. In such a case, an achieving strategy would be
\begin{equation}\label{OTFS_UV}
    \mathbf{V} = \mathbf{F}_M^H,\quad \mathbf{U} = \mathbf{I}_N,
\end{equation}
which is an OTFS waveform.
Completing the proof.
\end{proof}

Furthermore, a stronger result for the optimality of the OFDM and OTFS can be established based on the following theorem.
\begin{thm}\label{thm_OFDM_OTFS_OPT}
    OFDM and OTFS achieve the lowest ESL at each delay-Doppler bin over the FST-AF for sub- and super-Gaussian constellations, respectively.
\end{thm}
\begin{proof}
See Appendix \ref{proof_OFDM_OTFS_OPT}.
\end{proof}

\textbf{Remark 6}. We note that the above optimality of OFDM and OTFS can be regarded as a two-dimensional extension of the findings in \cite{liu2024ofdmachieveslowestranging}, which show that OFDM and SC yield the lowest average ranging sidelobe levels under sub-Gaussian and super-Gaussian constellations, respectively. The analogy stems from the fact that, within the FST-AF framework, OTFS can be seen as the Fourier dual of OFDM under the 2D signaling schemes, just as SC is the 1D Fourier dual of OFDM.

\section{Conclusions}\label{sec_5}
This paper investigated the ambiguity properties of random communication-centric ISAC waveforms through a unified treatment of the discrete periodic ambiguity function (DP-AF) and the fast-slow-time ambiguity function (FST-AF). By modeling signals generated from arbitrary orthonormal modulation bases with i.i.d. constellation symbols, we derived general closed-form expressions for the expected sidelobe level (ESL) and the expected integrated sidelobe level (EISL). These formulas were then specialized to four representative modulation schemes, namely, SC, OFDM, OTFS, and AFDM, to quantify how waveform structure and constellation statistics shape the ambiguity behavior.

For the DP-AF, we showed that the normalized EISL is invariant across all orthogonal modulation bases and constellations, as a direct consequence of the volume invariance of Woodward's classic AF for deterministic signals. Beyond this integrated metric, a group-theoretic analysis based on the finite-dimensional Weyl–Heisenberg representation revealed stronger structural limitations. In particular, under sub-Gaussian constellations, no orthogonal waveform can attain the minimum ESL over any compact two-dimensional delay-Doppler region. The lowest sidelobes can appear only along a one-dimensional cut, such as zero-delay or zero-Doppler slices, or at a finite set of well-separated delay-Doppler bins, with at most $N-1$ such minima. These results rule out both globally and locally optimal low-ambiguity zones in the DP-AF framework, even though the EISL itself is fixed.

Under the FST-AF formulation, the situation is more favorable for waveform optimization. The EISL and ESL now depend on both the modulation basis and the kurtosis of constellations. Within this regime, OFDM is optimal for sub-Gaussian constellations, while OTFS is optimal for super-Gaussian ones. Both waveforms minimize not only the EISL but also the ESL at every delay-Doppler bin. This two-dimensional duality between OFDM and OTFS parallels earlier one-dimensional ranging results for OFDM and SC, and highlights how Fourier-domain structure and symbol statistics jointly determine ambiguity performance.

Overall, the presented framework clarifies the fundamental limits and opportunities for using random communication waveforms for sensing. Future work may extend these methodologies to ISAC signaling schemes with nonlinear modulation and non-i.i.d. constellation symbols.

\appendices
\section{Proof of Proposition \ref{prop_ESL_DP_AF}} \label{proof of squared AF}
Let $\mathbf{B} = \mathbf{U}^H = \left[\mathbf{b}_1,\mathbf{b}_2,\ldots,\mathbf{b}_N\right]$, we have
\begin{align}\label{2D AF -2}
    &\nonumber \left| \mathcal{A}_{\rm{DP}}(k,q) \right|^2 
    = \left| \mathbf{x}^H \mathbf{J}_{N,k}^T\mathbf{D}_{N,q}^\ast  \mathbf{x} \right|^2 = \left| \mathbf{s}^H\mathbf{B} \mathbf{J}_{N,k}^T\mathbf{D}_{N,q}^\ast \mathbf{B}^H\mathbf{s} \right|^2\\
    & \nonumber = \left| \sum_{n=1}^N \mathbf{s}^H\mathbf{b}_{\langle N-k+n\rangle_N} \mathbf{b}_n^H\mathbf{s}\cdot e^{-j \frac{2\pi}{N} q (n-1)} \right|^2\\
    & = \sum_{n,m} \mathbf{b}_n^H \mathbf{s} \mathbf{s}^H \mathbf{b}_{\langle N-k+n\rangle_N}  \mathbf{b}_{\langle N-k+m\rangle_N}^H \mathbf{s} \mathbf{s}^H \mathbf{b}_m  e^{-j \frac{2\pi}{N} q (n-m)}.
\end{align}
By denoting $\tilde{\mathbf{s}} = \mathrm{vec}(\mathbf{s} \mathbf{s}^H) = \mathbf{s}^\ast \otimes \mathbf{s}$, it follows that
\begin{align}\label{expanding 1}
    &\mathbf{b}_n^H \mathbf{s} \mathbf{s}^H \mathbf{b}_{\langle N-k+n\rangle_N} = \left(\mathbf{b}_{\langle N-k+n\rangle _N}^{T} \otimes \mathbf{b}_n^H\right)\tilde{\mathbf{s}},\\
    &\mathbf{b}_{\langle N-k+m\rangle_N}^H \mathbf{s} \mathbf{s}^H \mathbf{b}_m = \tilde{\mathbf{s}}^H\left(\mathbf{b}_{\langle N-k+m\rangle _N}^{\ast} \otimes \mathbf{b}_m\right).
\end{align}
This implies that
\begin{equation}
\begin{aligned}
    &\mathbb{E}\left(\mathbf{b}_n^H \mathbf{s} \mathbf{s}^H \mathbf{b}_{\langle N-k+n\rangle_N}  \mathbf{b}_{\langle N-k+m\rangle_N}^H \mathbf{s} \mathbf{s}^H \mathbf{b}_m\right)\\
    &= \left(\mathbf{b}_{\langle N-k+n\rangle _N}^{T} \otimes \mathbf{b}_n^H\right)\mathbb{E}\left(\tilde{\mathbf{s}}{\tilde{\mathbf{s}}}^H\right)\left(\mathbf{b}_{\langle N-k+m\rangle _N}^{\ast} \otimes \mathbf{b}_m\right).
\end{aligned}
\end{equation}
By using \cite[Proposition 1]{liu2024ofdmachieveslowestranging}, $\mathbb{E}\left(\tilde{\mathbf{s}}{\tilde{\mathbf{s}}}^H\right)$ may be decomposed as
\begin{equation} \label{S}
        \mathbb{E}\left(\tilde{\mathbf{s}}{\tilde{\mathbf{s}}}^H\right) = \mathbf{I}_{N^2} + \mathbf{S}_1 + \mathbf{S}_2 ,
    \end{equation}
where 
    \begin{align}
        \mathbf{S}_1 &= \mathrm{Diag} \left( \left[\kappa - 2, \mathbf{0}_N^T, \kappa - 2, \mathbf{0}_N^T,..., \kappa - 2 \right]^T \right), \\
        \mathbf{S}_2 &= \left[\mathbf{h}, \mathbf{0}_{N^2 \times N},\mathbf{h},..., \mathbf{h}, \mathbf{0}_{N^2 \times N}, \mathbf{h} \right],
    \end{align}    
with $\mathbf{0}_{N^2\times N}$ being the all-zero matrix of size $N ^2 \times N$, and 
\begin{equation}
    \mathbf{h} = \left[ 1,\mathbf{0}_N^T,1,...,1,\mathbf{0}_N^T,1 \right]^T.
\end{equation}
By noting $\mathbf{b}_n^H \mathbf{b}_m = \delta_{m,n}$, we have
\begin{align} 
    \nonumber&(\mathbf{b}_{\langle N-k+n\rangle_N}^{T} \otimes \mathbf{b}_n^H) \mathbf{I}_{N^2} (\mathbf{b}_{\langle N-k+m\rangle_N}^{*} \otimes \mathbf{b}_m) \\
    &= \mathbf{b}_{\langle N-k+n \rangle _N}^{T} \mathbf{b}_{\langle N-k+m\rangle_N}^{*}  \mathbf{b}_n^H \mathbf{b}_m 
    = \delta_{m,n},  \label{term1} \\
    \nonumber&(\mathbf{b}_{\langle N-k+n \rangle _N}^{T} \otimes \mathbf{b}_n^H) \mathbf{S}_1 (\mathbf{b}_{\langle N-k+m \rangle _N}^{*} \otimes \mathbf{b}_m) \\
    &= (\kappa - 2)  \mathbf{1}^T (\mathbf{b}_{\langle N-k+n \rangle _N} \odot \mathbf{b}_n^{*} \odot \mathbf{b}_{\langle N-k+m \rangle _N}^{*} \odot \mathbf{b}_m) ,\label{term2} \\   
    \nonumber&(\mathbf{b}_{\langle N-k+n \rangle _N}^{T} \otimes \mathbf{b}_n^H) \mathbf{S}_2 (\mathbf{b}_{\langle N-k+m\rangle _N}^{*} \otimes \mathbf{b}_m) \\
    &= \mathbf{1}^T (\mathbf{b}_{\langle N-k+n \rangle _N} \odot \mathbf{b}_n^{*}) \cdot \mathbf{1}^T (\mathbf{b}_{\langle N-k+m \rangle _N}^{*} \odot \mathbf{b}_m)
    = \delta_{k,0} , \label{term3}   
\end{align}
Combing \eqref{term1}, \eqref{term2} and \eqref{term3}, we arrive at
\begin{align}\label{average sidelobe of AF}
    &\nonumber \mathbb{E}(\left| \mathcal{A}_{\rm{DP}}(k,q) \right|^2)
    =  \sum_{n,m} \delta_{m,n} e^{-j\frac{2\pi}{N}q(n-m)} + (\kappa - 2) \times\\
    &\nonumber \sum_{n,m} \mathbf{1}^T_N(\mathbf{b}_{\langle N-k+n\rangle_N} \odot \mathbf{b}_n^{*} \odot \mathbf{b}_{\langle N-k+m\rangle_N}^{*} \odot \mathbf{b}_m) e^{-j\frac{2\pi}{N}q(n-m)}  \\
    &+ \sum_{n,m}\delta_{k,0} e^{-j\frac{2\pi}{N}q(n-m)}.  
\end{align}
The first term of \eqref{average sidelobe of AF} is
\begin{align}\label{first term}
    \sum_{n,m}  \delta_{m,n} e^{-j\frac{2\pi}{N}q(n-m)} = N.
\end{align}
The third term of \eqref{average sidelobe of AF} is
\begin{align}\label{third term}
    \sum_{n,m}\delta_{k,0} e^{-j\frac{2\pi}{N}q(n-m)} = N^2 \delta_{k,0} \delta_{q,0}.
\end{align}
Furthermore, by noting $\mathbf{B} = \mathbf{U}^H$ and thereby $b_{\ell,n} = u_{n,\ell}^\ast$, one may obtain
\begin{align}\label{2nd_term}
    &\nonumber\sum_{n,m} \mathbf{1}^T_N(\mathbf{b}_{\langle N-k+n\rangle_N} \odot \mathbf{b}_n^{*} \odot \mathbf{b}_{\langle N-k+m\rangle_N}^{*} \odot \mathbf{b}_m) e^{-j\frac{2\pi}{N}q(n-m)}\\
    & \nonumber = \sum_{\ell=1}^N\left|\sum_{n=1}^Nb_{\ell,\langle N-k+n\rangle_N}b_{\ell,n}^\ast e^{-j\frac{2\pi}{N}qn}\right|^2\\
    & \nonumber = \sum_{\ell=1}^N\left|\sum_{n=1}^Nu_{\langle N-k+n\rangle_N,\ell}^\ast u_{n,\ell} e^{-j\frac{2\pi}{N}qn}\right|^2\\
    & = \sum_{\ell=1}^N\left|\mathbf{u}_\ell^H\mathbf{J}_{N,k}^T\mathbf{D}_{N,q}^\ast\mathbf{u}_\ell\right|^2 = \sum_{n=1}^N\left|\mathbf{u}_n^H\mathbf{D}_{N,q}\mathbf{J}_{N,k}\mathbf{u}_n\right|^2.
\end{align}
Substituting \eqref{2nd_term}, \eqref{third term}, and \eqref{first term} into \eqref{average sidelobe of AF} leads to \eqref{the average sidelobe}, which completes the proof.

\section{Proof of Theorem \ref{thm_no_LAZ}}\label{appendix_thm_1}
We prove this theorem using group-theoretic arguments. To proceed, we establish the following lemma.
\begin{lemma}
Let $\mathbf{U} \in \mathcal{U}(N)$ be an arbitrary unitary matrix. Then, the set $\mathcal{G}_\mathbf{U}$ constitutes an Abelian subgroup of $\mathcal{H}(N)$.
\end{lemma}
\begin{proof}
    Let $\mathbf{A}, \mathbf{B} \in \mathcal{G}_{\mathbf{U}}$. Since $\mathbf{U}^H \mathbf{A} \mathbf{U}$ and $\mathbf{U}^H \mathbf{B} \mathbf{U}$ are diagonal, they commute. Hence,
    \begin{align}
        \mathbf{A}\mathbf{B
        } &\nonumber = \mathbf{U}\mathbf{U}^H\mathbf{A}\mathbf{U}\mathbf{U}^H\mathbf{B}\mathbf{U}\mathbf{U}^H = \mathbf{U}\mathbf{U}^H\mathbf{B}\mathbf{U}\mathbf{U}^H\mathbf{A}\mathbf{U}\mathbf{U}^H\\&=\mathbf{B}\mathbf{A}
    \end{align}
    which proves that all elements in $\mathcal{G}_{\mathbf{U}}$ commute. Next, we verify the subgroup properties of $\mathcal{G}_{\mathbf{U}}$.
    \subsubsection{Closure}
    For any $\mathbf{A},\mathbf{B}\in\mathcal{G}_{\mathbf{U}}$,
    \begin{equation}
        \mathbf{U}^H\mathbf{A}\mathbf{B}\mathbf{U} = \mathbf{U}^H\mathbf{A}\mathbf{U}\mathbf{U}^H\mathbf{B}\mathbf{U},
    \end{equation}
    which is again diagonal. Since $\mathbf{A}\mathbf{B} \in \mathcal{H}(N)$, it follows that $\mathbf{A}\mathbf{B} \in \mathcal{G}_{\mathbf{U}}$.
    \subsubsection{Associativity}
    Associativity follows directly from the associativity of matrix multiplication.
    \subsubsection{Identity and Inverse}
    The identity element $\mathbf{I}_N\in\mathcal{G}(N)$ is diagonalized by any $\mathbf{U}$, and hence belongs to $\mathcal{G}_{\mathbf{U}}$. Moreover, note that all elements in $\mathcal{H}(N)$ are of full rank. Given $\mathbf{A}\in\mathcal{G}_{\mathbf{U}}$, we have
    \begin{equation}
        \mathbf{U}^H\mathbf{A}^{-1}\mathbf{U} = \left(\mathbf{U}^H\mathbf{A}\mathbf{U}\right)^{-1}
    \end{equation}
    is also a diagonal matrix. Therefore, for any element of $\mathcal{G}_{\mathbf{U}}$, its inverse element also belongs to $\mathcal{G}_{\mathbf{U}}$. 
    
    Thus, $\mathcal{G}_{\mathbf{U}}$ is a subgroup of $\mathcal{H}(N)$, and all of its elements commute. Therefore, $\mathcal{G}_{\mathbf{U}}$ is an Abelian subgroup of $\mathcal{H}(N)$.
\end{proof}
We now show that for any $(a,b) \in \mathbb{Z}_{N}^2$, it is impossible that
\begin{equation}\label{3_elements}
    (a,b), \quad(a,\langle b+1\rangle_N), \quad(\langle a+1 \rangle_N,b)
\end{equation}
belong to $\mathcal{I}_\mathbf{U}$ simultaneously. We prove this by contradiction. Suppose that the three elements all lie in $\mathcal{I}_\mathbf{U}$. Then the matrices in the set
\begin{equation}
    \mathcal{Q}:=\left\{e^{j\ell\frac{2\pi}{N}}\mathbf{G}_{a,b},\;\;e^{j\ell\frac{2\pi}{N}}\mathbf{G}_{a,\langle b+1\rangle_N},\;\;e^{j\ell\frac{2\pi}{N}}\mathbf{G}_{\langle a+1 \rangle_N,b}\right\}_{\ell\in\mathbb{Z}_N}
\end{equation}
are all diagonalizable by $\mathbf{U}$, implying that $\mathcal{Q}\subset\mathcal{G}_{\mathbf{U}}$. 

Since $\mathcal{G}_\mathbf{U}$ is a subgroup of $\mathcal{H}(N)$, it is closed under inversion. Note that 
\begin{align}
   &e^{-j(a+1)\frac{2\pi}{N}}\mathbf{G}_{a,b} = e^{j(N-a-1)\frac{2\pi}{N}}\mathbf{G}_{a,b}\in\mathcal{Q}\subset\mathcal{G}_{\mathbf{U}},\\
   &e^{-j\frac{2\pi}{N}} \mathbf{G}_{a,b} = e^{j(N-1)\frac{2\pi}{N}} \mathbf{G}_{a,b}\in\mathcal{Q}\subset\mathcal{G}_{\mathbf{U}}.
\end{align}
Therefore,
\begin{equation}
    e^{j(a+1)\frac{2\pi}{N}}\mathbf{G}_{a,b}^{-1}\in \mathcal{G}_\mathbf{U},\;\;e^{j\frac{2\pi}{N}} \mathbf{G}_{a,b}^{-1}\in \mathcal{G}_\mathbf{U}.
\end{equation}
By recalling the non-commutation relation \eqref{non_commutation_property}, for any elements $\mathbf{G}_{m,n}$ and $\mathbf{G}_{p,q}$, it holds that
\begin{align}
    \mathbf{G}_{m,n}\mathbf{G}_{p,q} = \mathbf{D}_{N,n}\mathbf{J}_{N,m}\mathbf{D}_{N,q}\mathbf{J}_{N,p} = e^{\frac{j2mq\pi}{N}}\mathbf{G}_{m+p,n+q}.
\end{align}
This indicates that
\begin{equation}
    e^{j\frac{2\pi}{N}}\mathbf{G}_{0,1} =  e^{j(a+1)\frac{2\pi}{N}}\mathbf{G}_{a,b}^{-1}\mathbf{G}_{a,\langle b+1\rangle_N}\in\mathcal{G}_{\mathbf{U}},
\end{equation}
and
\begin{equation}
     e^{j\frac{2\pi}{N}}\mathbf{G}_{1,0} =  e^{j\frac{2\pi}{N}} \mathbf{G}_{a,b}^{-1}\mathbf{G}_{\langle a+1\rangle_N,b}\in\mathcal{G}_{\mathbf{U}}.
\end{equation}
Since $e^{j\frac{2\pi}{N}}\mathbf{G}_{0,1},e^{j\frac{2\pi}{N}}\mathbf{G}_{1,0}$ generate $\mathcal{H}\left(N\right)$, and $\mathbf{G}_{0,0}\in\mathcal{G}_{\mathbf{U}}\subset\mathcal{H}\left(N\right)$ as the identity element, it follows that all the elements of $\mathcal{H}(N)$ can be obtained by the multiplications of $e^{j\frac{2\pi}{N}}\mathbf{G}_{0,1},e^{j\frac{2\pi}{N}}\mathbf{G}_{1,0}$, and $\mathbf{G}_{0,0}$. Hence one must have
\begin{equation}
    \mathcal{G}_{\mathbf{U}} = \mathcal{H}(N).
\end{equation}
However, $\mathcal{H}(N)$ is not Abelian, whereas $\mathcal{G}_{\mathbf{U}}$ is shown to be an Abelian group, leading to a contradiction. As a consequence, the three elements in \eqref{3_elements} cannot all belong to $\mathcal{I}_\mathbf{U}$ simultaneously. This completes the proof.

\section{Proof of Proposition \ref{prop_size_S}}\label{proof_prop_size_S}
Let us consider the space of all $N\times N$ complex matrices, denoted as $\mathbb{C}^{N\times N}$, which may be viewed as an $N^2$-dimensional Hilbert space equipped with the inner product:
\begin{equation}
    \langle\mathbf{A},\mathbf{B}\rangle = \operatorname{Tr}\left(\mathbf{A}^H\mathbf{B}\right),\quad \mathbf{A},\mathbf{B}\in\mathbb{C}^{N\times N}.
\end{equation}
It is straightforward to verify that, for all $\left(k,q\right),\left(m,n\right)\in\mathbb{Z}_{N}^2$,
\begin{equation}\label{G_orthogonality}
    \langle\mathbf{G}_{k,q},\mathbf{G}_{m,n}\rangle = \delta_{k,m}\delta_{q,n}.
\end{equation}
Hence, $\{\mathbf{G}_{k,q}\}_{(k,q)\in\mathbb{Z}_N^2}$ forms an orthogonal basis of $\mathbb{C}^{N\times N}$. Since $\mathbf{U}$ is unitary, it can be readily shown that
\begin{equation}
    \langle\mathbf{U}^H\mathbf{G}_{k,q}\mathbf{U},\mathbf{U}^H\mathbf{G}_{m,n}\mathbf{U}\rangle = \delta_{k,m}\delta_{q,n},\;\; (k,q),(m,n)\in\mathbb{Z}_N^2,
\end{equation}
which indicates that $\{\mathbf{U}^H\mathbf{G}_{k,q}\mathbf{U}\}_{(k,q)\in\mathbb{Z}_N^2}$ is also an orthogonal basis of $\mathbb{C}^{N\times N}$.

Let $\mathcal{D}(N)\subset\mathbb{C}^{N\times N}$ denote the subspace of all diagonal matrices, which has a dimension of $N$. By the definition of $\mathcal{G}_{\mathbf{U}}$, we have
\begin{equation}
\mathbf{U}^H\mathbf{G}_{k,q}\mathbf{U}\in \mathcal{D}(N), \quad (k,q)\in\mathcal{I}_{\mathbf{U}}.
\end{equation}
Moreover, it follows from \eqref{G_orthogonality} that
\begin{equation}
    \langle\mathbf{U}^H\mathbf{G}_{k,q}\mathbf{U},\mathbf{U}^H\mathbf{G}_{m,n}\mathbf{U}\rangle = \delta_{k,m}\delta_{q,n},\;\;(k,q),(m,n)\in\mathcal{I}_{\mathbf{U}}.
\end{equation}
Thus the set $\{\mathbf{U}^H\mathbf{G}_{k,q}\mathbf{U}\}_{(k,q)\in\mathcal{I}_{\mathbf{U}}}$ is a collection of mutually orthogonal non-zero vectors contained in $\mathcal{D}(N)$. In any $N$-dimensional inner product space, there can be at most $N$ orthogonal non-zero vectors. Therefore we have $\mathrm{Card\left(\mathcal{I}_{\mathbf{U}}\right)} = \mathrm{Card}\left(\{\mathbf{U}^H\mathbf{G}_{k,q}\mathbf{U}\}_{(k,q)\in\mathcal{I}_{\mathbf{U}}}\right)\leq N$.

\section{Proof of Proposition \ref{dispersion_properties}}\label{dispersion_proof}
By recalling the non-commutation relation \eqref{non_commutation_property}, for any elements $\mathbf{G}_{m,n}$ and $\mathbf{G}_{p,q}$, it holds that
\begin{align}
    &\mathbf{G}_{m,n}\mathbf{G}_{p,q}  = e^{\frac{j2mq\pi}{N}}\mathbf{G}_{m+p,n+q},\\
    &\mathbf{G}_{p,q}\mathbf{G}_{m,n}  =  e^{\frac{j2np\pi}{N}}\mathbf{G}_{m+p,n+q}.
\end{align}
This implies that $\mathbf{G}_{m,n}$ and $\mathbf{G}_{p,q}$ commute if and only if the symplectic form
\begin{equation}\label{sympletic_form}
    \omega\left\{(m,n),(p,q)\right\} = \langle np-mq\rangle_{N} = 0.
\end{equation} 

Let $\mathbf{a} = \left(a_1,a_2\right),\mathbf{b} = \left(b_1,b_2\right),\mathbf{c} = \left(c_1,c_2\right)$ be three elements in the index set $\mathcal{I}_{\mathbf{U}}$. By the closure of $\mathcal{G}_{\mathbf{U}}$, it holds immediately that $\mathbf{b}-\mathbf{a}\in\mathcal{I}_{\mathbf{U}}$ and $\mathbf{c}-\mathbf{a}\in\mathcal{I}_{\mathbf{U}}$. Recalling that $\mathcal{G}_{\mathbf{U}}$ is Abelian, the points $\mathbf{b}-\mathbf{a}$ and $\mathbf{c}-\mathbf{a}$ must satisfy \eqref{sympletic_form}. Thus, the area of the triangle formed by the three points $\mathbf{a},\mathbf{b},\mathbf{c}$ may be computed by
\begin{equation}
    \Delta_{\mathbf{abc}} = \frac{1}{2}\left|\omega\left(\mathbf{b}-\mathbf{a},\mathbf{c}-\mathbf{a}\right)\right| = \frac{kN}{2}, \quad k\in\mathbb{Z}_N.
\end{equation}
Hence, if $k = 0$, the three points are collinear. Otherwise, they form a triangle with an area at least $\frac{N}{2}$.

\section{Proof of Proposition \ref{the average squared DP-AF of OTFS}} \label{the proof of OTFS}
For OTFS waveform, its modulation basis $\mathbf{U} = \mathbf{F}_{N_1}^H \otimes \mathbf{I}_{N_2}$ may be expanded as
\begin{equation}\label{the U of OTFS}
\begin{aligned}
     \mathbf{U} 
    = \frac{1}{\sqrt{N_1}} \left[ \begin{matrix}
        \mathbf{I}_{N_2} & \mathbf{I}_{N_2} & \dots & \mathbf{I}_{N_2} \\
        \mathbf{I}_{N_2} & e^{j \frac{2\pi}{N_1}} \mathbf{I}_{N_2} & \dots & e^{j \frac{2\pi}{N_1}({N_1}-1)} \mathbf{I}_{N_2} \\
        \vdots&\vdots&\vdots&\vdots\\
        \mathbf{I}_{N_2} & e^{j \frac{2\pi}{N_1}({N_1}-1)} \mathbf{I}_{N_2} & \dots & e^{j \frac{2\pi}{N_1}({N_1}-1)^2} \mathbf{I}_{N_2} \\
    \end{matrix} \right],  
\end{aligned}
\end{equation}
where
\begin{equation}\label{column of OTFS}
\begin{aligned}
     &\mathbf{u}_n = \frac{1}{\sqrt{N_1}}\mathbf{J}_{N,\langle n-1\rangle_{N_2}} \times\\
     & \left[ 1, \mathbf{0}_{N_2-1}^T, e^{j \frac{2 \pi}{N_1}(v-1)}, \mathbf{0}_{N_2-1}^T, ..., e^{j \frac{2 \pi}{N_1}(v-1)(N_1-1)}, \mathbf{0}_{N_2-1}^T\right]^T,   
\end{aligned}
\end{equation}
with $v = \lceil \frac{n}{N_2} \rceil$. Furthermore, note that
\begin{align}\label{otfs_waveform_term1}
    &\nonumber \sum_{n=1}^N \left| \mathbf{u}_n^H \mathbf{D}_{N,q} \mathbf{J}_{N,k} \mathbf{u}_n\right|^2 \\
    & = N \mathbf{f}_{N,q+1}^H \left[ \sum_{n=1}^N \mathrm{Diag}(\mathbf{u}_n^\ast)  \mathbf{J}_{N,k} \mathbf{u}_n \mathbf{u}_n^H \mathbf{J}_{N,k}^T \mathrm{Diag}(\mathbf{u}_n) \right] \mathbf{f}_{N,q+1}.
\end{align}   
According to the structure of $\mathbf{u}_n$, we have
\begin{equation}\label{Diag_J_un}
\begin{aligned}
    \mathrm{Diag}(\mathbf{u}_n^\ast)  \mathbf{J}_{N,k}
     \mathbf{u}_n  = e^{-j \frac{2\pi}{N_1} (v-1)\frac{k}{N_2}} \mathrm{Diag}(\mathbf{u}_n^\ast) 
     \mathbf{u}_n \delta_{\langle k \rangle_{N_2},0}.
\end{aligned}
\end{equation}
Plugging \eqref{Diag_J_un} into \eqref{otfs_waveform_term1} yields
\begin{align}
    &\nonumber \mathrm{Diag}(\mathbf{u}_n^\ast)  \mathbf{J}_{N,k} \mathbf{u}_n \mathbf{u}_n^H \mathbf{J}_{N,k}^T \mathrm{Diag}(\mathbf{u}_n) \\
    & = \mathrm{Diag}(\mathbf{u}_n^\ast) \mathbf{u}_n \mathbf{u}_n^H  \mathrm{Diag}(\mathbf{u}_n) \delta_{\langle k \rangle_{N_2},0},
\end{align}
where
\begin{align}
     &\nonumber \mathrm{Diag}(\mathbf{u}_n^\ast) \mathbf{u}_n \\
     &\nonumber= \frac{1}{N_1}\mathbf{J}_{N,\langle n-1\rangle_{N_2}} \left[ 1, \mathbf{0}_{N_2-1}^T, 1, \mathbf{0}_{N_2-1}^T, ..., 1, \mathbf{0}_{N_2-1}^T\right]^T \\
     &\nonumber = \frac{1}{N_1}\mathbf{J}_{N,\langle n-1\rangle_{N_2}} \left( \mathbf{1}_{N_1} \otimes [1, \mathbf{0}_{N_2-1}^T]^T \right)\\
     & = \frac{1}{N_1} \left( \mathbf{1}_{N_1} \otimes \mathbf{J}_{N_2,\langle n-1\rangle_{N_2}}[1, \mathbf{0}_{N_2-1}^T]^T \right).
\end{align}
Accordingly, the following result follows directly
\begin{align}
    &\nonumber \mathrm{Diag}(\mathbf{u}_n^\ast) \mathbf{u}_n \mathbf{u}_n^H  \mathrm{Diag}(\mathbf{u}_n)\\
    & = \frac{1}{N_1^2}\mathbf{1}_{N_1}\mathbf{1}_{N_1}^T \otimes \mathrm{Diag}\left( \mathbf{J}_{N_2,\langle n-1\rangle_{N_2}}[1, \mathbf{0}_{N_2-1}^T]^T \right).
\end{align}
It then follows that the summation term is given by
\begin{equation}\label{otfs_waveform_term2}
\begin{aligned}
    &\sum_{n=1}^N \mathrm{Diag}(\mathbf{u}_n^\ast)  \mathbf{J}_{N,k}
     \mathbf{u}_n\mathbf{u}_n^H \mathbf{J}_{N,k}^T \mathrm{Diag}(\mathbf{u}_n)\\
    &= \frac{1}{N_1} (\mathbf{1}_{N_1,N_1} \otimes \mathbf{I}_{N_2})  \delta_{\langle k \rangle_{N_2},0}. 
\end{aligned}
\end{equation}
Plugging \eqref{otfs_waveform_term2} into \eqref{otfs_waveform_term1} yields
\begin{align}\label{otfs_waveform_term3}
    &\nonumber N \mathbf{f}_{N,q+1}^H \left[ \sum_{n=1}^N \mathrm{Diag}(\mathbf{u}_n^\ast)  \mathbf{J}_{N,k} \mathbf{u}_n \mathbf{u}_n^H \mathbf{J}_{N,k}^T \mathrm{Diag}(\mathbf{u}_n) \right] \mathbf{f}_{N,q+1} \\
    &=\frac{N}{N_1} \mathbf{f}_{N,q+1}^H \left(\mathbf{1}_{N_1,N_1} \otimes \mathbf{I}_{N_2} \right)   \mathbf{f}_{N,q+1} \delta_{\langle k \rangle_{N_2},0},
\end{align}
where the $n$th element of $\mathbf{f}_{N,q+1}^H  (\mathbf{1}_{N_1}\mathbf{1}_{N_1}^T \otimes \mathbf{I}_{N_2})$ is calculated as
\begin{align}
    &\nonumber \frac{e^{j \frac{2\pi}{N}q \langle n-1 \rangle_{N_2}}}{\sqrt{N}}\sum_{m=1}^{N_1} e^{j \frac{2\pi}{N}q(m-1)N_2}  = \frac{N_1 e^{j \frac {2\pi}{N}q \langle n-1 \rangle_{N_2}} \delta_{{\langle q \rangle}_{N_1},0}}{\sqrt{N}}.
\end{align}
Thus, the value of $\mathbf{f}_{N,q+1}^H  (\mathbf{1}_{N_1,N_1} \otimes \mathbf{I}_{N_2})\mathbf{f}_{N,q+1}$ is
\begin{align}\label{otfs_waveform_term4}
    &\nonumber \sum_{n=1}^N \frac{N_1}{\sqrt{N}} e^{j \frac {2\pi}{N}q \langle n-1 \rangle_{N_2}} \delta_{{\langle q \rangle}_{N_1},0} \times \frac{1}{\sqrt{N} } e^{-j \frac {2\pi}{N}q (n-1)}\\
    &\nonumber = \sum_{n=1}^N \frac{1}{N_2} e^{j \frac {2\pi}{N}q \langle n-1 \rangle_{N_2}} e^{-j \frac {2\pi}{N}q (n-1)} \delta_{{\langle q \rangle}_{N_1},0}\\
    &= N_1 \delta_{{\langle q \rangle}_{N_1},0}.
\end{align}
Combing \eqref{otfs_waveform_term4}, \eqref{otfs_waveform_term3}, and \eqref{the average sidelobe} leads to \eqref{DP_AF_OTFS}.

\section{Proof of Proposition \ref{afdm_sp_af_squared}} \label{the proof of AFDM}
Let $\mathbf{U} = \mathbf{\Lambda}_{c_1}^H \mathbf{F}_N^H \mathbf{\Lambda}_{c_2}^H$, the $n$-th column of $\mathbf{U}$ is
\begin{equation}\label{column of AFDM}
    \mathbf{u}_n = e^{j 2 \pi c_2 (n-1)^2}\mathbf{\Lambda}_{c_1}^H \mathbf{f}_{N,n}^\ast.
\end{equation}
Plugging \eqref{column of AFDM} into \eqref{the average sidelobe}, we can arrive at
    \begin{align}\label{AFDM-addition-proof}
        \nonumber&\sum_{n=1}^N \left| \mathbf{u}_n^H \mathbf{D}_{N,q} \mathbf{J}_{N,k} \mathbf{u}_n \right|^2 \\
        \nonumber&=\sum_{n=1}^N \left| \mathbf{f}_{N,i}^T \mathbf{\Lambda}_{c_1}  \mathbf{D}_{N,q} \mathbf{J}_{N,k} \mathbf{\Lambda}_{c_1}^H \mathbf{f}_{N,i}^\ast \right|^2\\
        \nonumber& =\sum_{n=1}^N \left|(\mathbf{c}_1 \odot \mathbf{f}_{N,i})^T  \mathbf{D}_{N,q} \mathbf{J}_{N,k} (\mathbf{c}_1 \odot \mathbf{f}_{N,i})^\ast \right|^2\\
        & =\sum_{n=1}^N \left|\mathbf{d}_n^T  \mathbf{D}_{N,q} \mathbf{J}_{N,k} \mathbf{d}_n^\ast\right|^2 =\sum_{n=1}^N \left| \mathbf{d}_n^H  \mathbf{D}_{N,q}^\ast \mathbf{J}_{N,k} \mathbf{d}_n\right|^2,
    \end{align}
where $\mathbf{d}_n = \mathbf{c}_n\odot \mathbf{f}_n$, and the $m$-th element of $\mathbf{d}_n$ is
\begin{equation}\label{the element of d}
\begin{aligned}
    d_{m,n} 
     = \frac{1}{\sqrt{N}} e^{-j \frac{2 \pi}{N} [(n-1) (m-1) + \frac{p}{2}(m-1)^2]},
\end{aligned}
\end{equation}
where $p = 2Nc_1 \in\mathbb{Z}$. Furthermore, \eqref{AFDM-addition-proof} can be simplified to \eqref{calculation of di fq}, where $\lambda \in\mathbb{Z}_N$.
\begin{figure*}[htbp]
\normalsize
\newcounter{MYtempeqncnt3}
\setcounter{MYtempeqncnt3}{\value{equation}}
\setcounter{equation}{151}
\begin{align}\label{calculation of di fq}
    &\nonumber \sum_{n=1}^N \left| \mathbf{d}_n^H  \mathrm{Diag}(\mathbf{f}_{N,q+1}) \mathbf{J}_{N,k} \mathbf{d}_n \right|^2  = \sum_{n=1}^N \left| \sum_{m=1}^N f_{m,q+1} d_{m,n}^\ast d_{\langle m-k \rangle_N,n} \right|^2\\     
    &\nonumber = \frac{1}{N^3} \sum_{n=1}^N \left|  \sum_{m=1}^N  e^{-j \frac{2 \pi}{N} q (m-1)} e^{j \frac{2 \pi}{N} [(n-1) (m-1) + \frac{p}{2}(m-1)^2]}  e^{-j \frac{2 \pi}{N} [(n-1) (\langle m-k \rangle_N-1) + \frac{p}{2}(\langle m-k \rangle_N-1)^2]}\right|^2\\
    &\nonumber = \frac{1}{N^3} \sum_{n=1}^N \left| \sum_{m=1}^N e^{-j \frac{2 \pi}{N} q (m-1)} \cdot e^{j \frac{2 \pi}{N} \left[(n-1) (m-\langle m-k \rangle_N) \right]} \cdot e^{ j \frac{2 \pi}{N} \left[ \frac{p}{2}(m-1)^2 - \frac{p}{2}(\langle m-k \rangle_N-1)^2 \right]} \right|^2\\
    &\nonumber = \frac{1}{N^3} \sum_{n=1}^N \left| \sum_{m=1}^N e^{-j \frac{2 \pi}{N} q (m-1)} \cdot e^{ j \frac{\pi}{N} p \left[ (m-1)^2 - (m-k+\lambda N-1)^2 \right]} \right|^2  = \frac{1}{N^2} \left| \sum_{m=1}^N  e^{-j \frac{2 \pi}{N} q (m-1)} \cdot e^{ -j \frac{2\pi}{N} (m-1) \left[(2Nc_1 (\lambda N-k) \right]} \right|^2\\
    &= \frac{1}{N^2} \left| \sum_{m=1}^N e^{ -j \frac{2\pi}{N} (m-1) \left[(2Nc_1 (\lambda N-k) +q\right]} \right|^2 = \delta_{\langle 2Nkc_1 -q\rangle_N,0}.
\end{align}
\vspace*{3pt} 
\hrulefill 
\end{figure*}
Substituting \eqref{calculation of di fq} into \eqref{the average sidelobe} leads to the average squared DP-AF of AFDM in \eqref{dp_af_afdm}.

\section{Proof of Proposition \ref{prop_EISL_FST_AF}} \label{proof of the EISL of FST-AF}
For convenience, let us denote
\setcounter{equation}{152}
\begin{equation}
     \operatorname{vec}\left(\mathbf{X}_{\rm FT}\right) = \left(\mathbf{V}^T\otimes\mathbf{F}_N\mathbf{U}\right){\mathbf{s}} \triangleq \mathbf{Q}^H{\mathbf{s}},
\end{equation}
where ${\mathbf{s}} = \operatorname{vec}\left(\mathbf{S}\right)\in\mathbb{C}^{NM}$, and $\mathbf{Q} = \left[{\mathbf{q}}_1,{\mathbf{q}}_2,\ldots,{\mathbf{q}}_{NM}\right]$ is a size-$NM$ unitary matrix. By recalling \eqref{FST_AF_volume}, we have
\begin{align}
     &\mathbb{E}\left(\left\|\mathbf{A}_{\rm{FST}}\right\|_F^2\right) \nonumber = MN\mathbb{E}\left(\left\|\mathbf{F}_N\mathbf{U}\mathbf{S}\mathbf{V}\right\|_4^4\right)\\ & =MN\mathbb{E}\left(\left\|\mathbf{Q}^H{\mathbf{s}}\right\|_4^4\right)
    = MN\sum\limits_{n=1}^{NM} \mathbb{E}\left(\left|\mathbf{q}_n^H{\mathbf{s}}\right|^4\right).
\end{align}
Expanding $\left|\mathbf{q}_n^H{\mathbf{s}}\right|^2$ yields
\begin{align}
    \left|\mathbf{q}_n^H{\mathbf{s}}\right|^2 = \mathbf{q}_n^H{\mathbf{s}} {\mathbf{s}}^H \mathbf{q}_n \nonumber&= (\mathbf{q}_n^T \otimes \mathbf{q}_n^H) \operatorname{vec}({\mathbf{s}} {\mathbf{s}}^H) \\
    & = \operatorname{vec}^H ({\mathbf{s}} {\mathbf{s}}^H) (\mathbf{q}_n^\ast \otimes \mathbf{q}_n).
\end{align}
Accordingly,
\begin{align}
    &\mathbb{E}\left(\left\|\mathbf{A}\right\|_F^2\right) = MN\sum\limits_{n=1}^{NM} \left(\mathbf{q}_n^T\otimes\mathbf{q}_n^H\right)\tilde{\mathbf{S}} \left(\mathbf{q}_n^\ast\otimes\mathbf{q}_n\right),
\end{align}
where $\tilde{\mathbf{S}} = \mathbb{E}\left\{\operatorname{vec}\left({\mathbf{s}}{\mathbf{s}}^H\right)\operatorname{vec}^H\left({\mathbf{s}}{\mathbf{s}}^H\right)\right\} $, which can be again decomposed as \cite{liu2024ofdmachieveslowestranging}
    \begin{equation} \label{S}
        \tilde{\mathbf{S}}  = \mathbf{I}_{M^2 N^2} + \tilde{\mathbf{S}}_1 + \tilde{\mathbf{S}}_2.
    \end{equation}
Here, 
    \begin{align}
        \tilde{\mathbf{S}}_1 &= \mathrm{Diag} \left( \left[\kappa - 2, \mathbf{0}_{MN}^T, \kappa - 2, \mathbf{0}_{MN}^T,..., \kappa - 2 \right]^T \right), \\
        \tilde{\mathbf{S}}_2 &= \left[\tilde{\mathbf{h}}, \mathbf{0}_{M^2N^2 \times MN},\tilde{\mathbf{h}},..., \tilde{\mathbf{h}}, \mathbf{0}_{M^2N^2 \times MN}, \tilde{\mathbf{h}} \right],
    \end{align}    
where 
    \begin{equation}
        \tilde{\mathbf{h}} = \left[ 1,\mathbf{0}_{MN}^T,1,...,1,\mathbf{0}_{MN}^T,1 \right]^T.
    \end{equation}
By leveraging the same technique as in Appendix \ref{proof of squared AF}, we have
\begin{align}
\left(\mathbf{q}_n^T\otimes\mathbf{q}_n^H\right)\tilde{\mathbf{S}}\left(\mathbf{q}_n^\ast\otimes\mathbf{q}_n\right) = 2 + \left(\kappa - 2\right)\left\|\mathbf{q}_n\right\|_4^4.
\end{align}
It follows that
\begin{align}\label{E_A_FST}
    \mathbb{E}\left(\left\|\mathbf{A}_{\rm{FST}}\right\|_F^2\right) = 2M^2N^2 + \left(\kappa - 2\right)MN\left\|\mathbf{Q}\right\|_4^4.
\end{align}
Substituting \eqref{E_A_FST} and the average mainlobe level \eqref{FST_mainlobe} into \eqref{EISL_FST_definition} yields \eqref{the EISL of stop-and-go model}.

\section{Proof of Proposition \ref{prop_ESL_FST_AF}} \label{proof ot the average sidelobe of RDM}
According to \eqref{A_FST}
\begin{align}
    a_{n,m} &\nonumber = \sqrt{MN}\mathbf{f}_{N,n}^H\left|\mathbf{F}_N\mathbf{U}\mathbf{S}\mathbf{V}\right|^2\mathbf{f}_{M,m}\\
    &\nonumber = \sqrt{MN}\left(\mathbf{f}_{M,m}^T\otimes\mathbf{f}_{N,n}^H\right)\operatorname{vec}\left(\left|\mathbf{F}_N\mathbf{U}\mathbf{S}\mathbf{V}\right|^2\right)\\
    & \nonumber= \sqrt{MN}\left(\mathbf{f}_{M,m}^T\otimes\mathbf{f}_{N,n}^H\right)\left|\operatorname{vec}\left(\mathbf{F}_N\mathbf{U}\mathbf{S}\mathbf{V}\right)\right|^2 \\
    & = \sqrt{MN}\left(\mathbf{f}_{M,m}^T\otimes\mathbf{f}_{N,n}^H\right)\left|\mathbf{Q}^H{{\mathbf{s}}}\right|^2.
\end{align}
It follows that
\begin{align}
    \nonumber \mathbb{E}\left(\left|a_{n,m}\right|^2\right) &= MN\left(\mathbf{f}_{M,m}^T\otimes\mathbf{f}_{N,n}^H\right) \cdot\\& \mathbb{E}\left(\left|\mathbf{Q}^H{{\mathbf{s}}}\right|^2{\left|\mathbf{Q}^H{{\mathbf{s}}}\right|^2}^H\right)\left(\mathbf{f}_{M,m}^\ast\otimes\mathbf{f}_{N,n}\right).
\end{align}
Now let us analyze the matrix $\mathbb{E}\left(\left|\mathbf{Q}^H{{\mathbf{s}}}\right|^2{\left|\mathbf{Q}^H{{\mathbf{s}}}\right|^2}^H\right)$, whose $\left(l,k\right)$-th element reads
\begin{align}
&\nonumber\mathbb{E}\left(\left|\mathbf{q}_l^H{{\mathbf{s}}}\right|^2{\left|\mathbf{q}_k^H{{\mathbf{s}}}\right|^2}^H\right) = (\mathbf{q}_l^T\otimes \mathbf{q}_l^H)\tilde{\mathbf{S}}(\mathbf{q}_k^\ast\otimes \mathbf{q}_k)\\
&= \delta_{l,k} + 1 + \left(\kappa - 2\right)\left\|\mathbf{q}_l\odot\mathbf{q}_k\right\|^2.
\end{align}
Let us denote $\mathbf{w} = \sqrt{MN}\left(\mathbf{f}_{M,m}^\ast\otimes\mathbf{f}_{N,n}\right)$, and represent its $k$-th entry as $w_k$, we have
\begin{align}
    &\nonumber \mathbb{E}\left(\left|a_{n,m}\right|^2\right) = \sum_{l,k}^{NM}\left(\delta_{l,k} + 1 + \left(\kappa - 2\right)\left\|\mathbf{q}_l\odot\mathbf{q}_k\right\|^2\right)w_l w_k^\ast\\
    &\nonumber = \sum_{l=1}^{NM}\left|w_l\right|^2 + \left|\sum_{l=1}^{NM}w_l\right|^2 + \sum_{l,k}^{NM}\left(\kappa - 2\right)\left\|\mathbf{q}_l\odot\mathbf{q}_k\right\|^2 w_l w_k^\ast\\
    & = MN + \left|\mathbf{w}^H\mathbf{1}_{NM}\right|^2 + \left(\kappa - 2\right)\left\|\left|\mathbf{Q}\right|^2\mathbf{w}\right\|^2.
\end{align}
It can be readily shown
\begin{align}
    \mathbf{w}^H\mathbf{1}_{NM} &\nonumber = \sqrt{MN}\mathbf{f}_{N,n}^H\mathbf{1}_N\mathbf{1}_M^T\mathbf{f}_{M,m} = NM\delta_{1,n}\delta_{1,m},
\end{align}
which suggests that
\begin{align}
    \mathbb{E}\left(\left|a_{n,m}\right|^2\right) &= \nonumber MN + M^2N^2\delta_{1,n}\delta_{1,m} \\ & + \left(\kappa - 2\right)MN\left\|\left|\mathbf{Q}\right|^2\left(\mathbf{f}_{M,m}^\ast\otimes\mathbf{f}_{N,n}\right)\right\|^2,
\end{align}
where $n = 1,2,\cdots,N$ and $m = 1,2,\cdots,M$. By noting that
\begin{align}
    &\nonumber \left\| \left|\mathbf{Q}\right|^2\left(\mathbf{f}_{M,m}^\ast\otimes\mathbf{f}_{N,n}\right) \right\|^2 = \left\| \left| \mathbf{V}^\ast \otimes \mathbf{U}^H \mathbf{F}_N^H \right|^2\left(\mathbf{f}_{M,m}^\ast\otimes\mathbf{f}_{N,n}\right) \right\|^2\\
    &\nonumber = \left\| \left( \left| \mathbf{V}^\ast \right|^2 \otimes \left|\mathbf{U}^H \mathbf{F}_N^H \right|^2\right) \left(\mathbf{f}_{M,m}^\ast\otimes\mathbf{f}_{N,n}\right) \right\|^2\\
    & = \left\| \left| \mathbf{V} \right|^2 \mathbf{f}_{M,m}^\ast \right\|^2 \left\| \left|\mathbf{U}^H \mathbf{F}_N^H \right|^2 \mathbf{f}_{N,n} \right\|^2,
\end{align}
the average squared FST-AF maybe rearranged as \eqref{the average sidelobe of stop-and-go model}.


\section{Proof of Proposition \ref{FST_AF_AFDM_prop}} \label{proof of AFDM of 'stop-and-go' model}
By plugging the AFDM basis $\mathbf{U} = \mathbf{\Lambda}_{c_1}^{H} \mathbf{F}_N^H \mathbf{\Lambda}_{c_2}^H$ and $\mathbf{V} = \mathbf{I}_M$ into \eqref{the average sidelobe of stop-and-go model}, we have
\begin{align}
    &\nonumber \left\| \left| \mathbf{V} \right|^2 \mathbf{f}_{M,q+1}^\ast \right\|^2 \left\| \left|\mathbf{U}^H \mathbf{F}_N^H \right|^2 \mathbf{f}_{N,k+1} \right\|^2\\
    &\nonumber = \left\| \mathbf{f}_{M,q+1}^\ast \right\|^2 \left\|  \left| \mathbf{\Lambda}_{c_2} \mathbf{F}_N \mathbf{\Lambda}_{c_1} \mathbf{F}_N^H\right|^2 \mathbf{f}_{N,k+1} \right\|^2\\
    &\nonumber = \left\|  \left| \mathbf{\Lambda}_{c_2} \mathbf{F}_N \mathbf{\Lambda}_{c_1} \mathbf{F}_N^H\right|^2 \mathbf{f}_{N,k+1} \right\|^2\\
    &= \left\|  \left| \mathbf{F}_N \mathbf{\Lambda}_{c_1} \mathbf{F}_N^H\right|^2 \mathbf{f}_{N,k+1} \right\|^2.
\end{align}
The average squared FST-AF for AFDM is
\begin{align}\label{the average sidelobe of AFDM ''stop-and-go''}
    &\nonumber\mathbb{E}\left( \left| \mathcal{A}_{\rm{FST}}^{\rm{AFDM}} (k,q) \right|^2 \right) = M^2N^2\delta_{k,0}\delta_{q,0}+ MN\\
    &\quad\quad\quad\quad+ \left(\kappa - 2\right)MN  \left\|  \left| \mathbf{F}_N \mathbf{\Lambda}_{c_1} \mathbf{F}_N^H\right|^2 \mathbf{f}_{N,k+1} \right\|^2.
\end{align}
To proceed, let us discuss the cases of $c_1 = 0$ and $c_1\ne 0$, respectively.
\subsubsection{$c_1 = 0$}
The EISL of the FST-AF for AFDM  is
\begin{align}
    &\nonumber M^2 N^2 - MN + (\kappa -2)MN(MN-1) \\
    &= MN(MN-1)(\kappa -1).    
\end{align}
Accordingly, the ESL of AFDM is
    \begin{equation}
    \begin{aligned}
    \mathbb{E}\left( \left| \mathcal{A}_{\rm{FST}}^{\rm{AFDM}} (k,q) \right|^2 \right)
    &= (\kappa -1)MN, \quad k,q \ne 0,
    \end{aligned}
    \end{equation}
which is equivalent to those of the OFDM. This is because AFDM reduces to OFDM when $c_1 = 0$.

\subsubsection{$c_1 \ne 0$}
Let $\mathbf{Z} = \mathbf{F}_N \mathbf{\Lambda}_{c_1} \mathbf{F}_N^H$. The first row of $\mathbf{\left| Z \right|}^2$ can be calculated as
\begin{equation}\label{the first row of Z^2}
    \frac{1}{N} \left[ \begin{matrix}
         \left|\mathbf{f}_{N,1}^H \mathbf{c}_1\right|^2, \left|\mathbf{f}_{N,2}^H \mathbf{c}_1\right|^2 , ... , \left|\mathbf{f}_{N,N}^H \mathbf{c}_1\right|^2
    \end{matrix}\right],
\end{equation}
Since $\left| \mathbf{Z} \right|^2$ is also circulant, it can be decomposed as
\begin{equation}
    \left| \mathbf{Z} \right|^2 = \left| \mathbf{F}_N \mathbf{\Lambda}_{c_1} \mathbf{F}_N^H \right|^2 =  \mathbf{F}_N \mathrm{Diag}(\tilde{\mathbf{z}}) \mathbf{F}_N^H.
\end{equation}
Consequently we have
\begin{equation}
\begin{aligned}
    \left\|  \left| \mathbf{Z} \right|^2 \mathbf{f}_{N,k+1} \right\|^2
    = \left\|  \mathbf{F}_N \mathrm{Diag}(\tilde{\mathbf{z}}) \mathbf{F}_N^H \mathbf{f}_{N,k+1} \right\|^2 
    = |\tilde{z}_{k+1}|^2,    
\end{aligned}
\end{equation}
where $k \in\mathbb{Z}_N$. According to \eqref{the first row of Z^2}, $\tilde{\mathbf{z}}$ may be written as
\begin{align}
    \tilde{\mathbf{z}}& \nonumber= \sqrt{N} \mathbf{F}_N  \frac{1}{N} \left[\begin{matrix}
        \left|\mathbf{f}_{N,1}^H \mathbf{c}_1\right|^2,
        \left|\mathbf{f}_{N,2}^H \mathbf{c}_1\right|^2, ...,
        \left|\mathbf{f}_{N,N}^H \mathbf{c}_1\right|^2
    \end{matrix}\right]^T \\
    &\nonumber = \frac{1}{\sqrt{N}}  \mathbf{F}_N \left\{ (\mathbf{F}_N^H \mathbf{c}_1) \odot (\mathbf{F}_N^H \mathbf{c}_1)^\ast \right\}
\end{align}
which suggests that $\tilde{\mathbf{z}}$ is the periodic auto-correlation of $\mathbf{c}_1$. Therefore, the $k$th element of $\tilde{\mathbf{z}}$ is 
\begin{equation}
    \tilde{z}_{k+1} = \frac{1}{N} \mathbf{c}_1^T \mathbf{J}_{N,k} \mathbf{c}_1^\ast, \quad k\in\mathbb{Z}_N.
\end{equation}
Consequently,
\begin{align}\label{the value of |yn|^2}
    &\nonumber \left|\tilde{z}_{k+1} \right|^2 = \frac{1}{N^2} \left| \sum_{n=1}^N c_n c_{\langle n -k \rangle_N}^\ast \right|^2\\
    &\nonumber = \frac{1}{N^2} \left| \sum_{i=1}^N e^{-j 2 \pi c_1 \left[ (n-1)^2 - (\ell N + n -k -1)^2 \right]} \right|^2\\
    &\nonumber = \frac{1}{N^2} \left| \sum_{n=1}^N e^{j 2 \pi c_1 \left[ 2(n-1)(\ell  N - k) + (\ell  N - k)^2 \right]} \right|^2\\
    &\nonumber = \frac{1}{N^2} \left| \sum_{n=1}^N e^{j \frac{2 \pi}{N} \left[ 2N c_1 (n-1)(\ell  N - k)\right]} \right|^2\\
    & = \delta _{ \langle 2Nc_1 k \rangle_{N}, 0},
\end{align}
where $\ell\in\mathbb{Z}$ is an arbitrary integer. Taking \eqref{the value of |yn|^2} into \eqref{the average sidelobe of AFDM ''stop-and-go''}, we arrive at \eqref{FST_AF_AFDM}.

To evaluate the EISL of the FST-AF for AFDM, observe that the ESL equals $\left(\kappa - 1\right)MN$ whenever $ \delta _{ \langle 2Nc_1 k \rangle_{N}, 0} = 1$, and equals $MN$ otherwise. Let $\phi \in \mathbb{Z}$ denote $\gcd\left(2Nc_1, N\right)$. Then
\begin{align}
    \delta _{ \langle 2Nc_1 k \rangle_{N}, 0} = 1, \text{if and only if}\;\;k = \frac{nN}{\phi},
\end{align}
where $n = 0,1,\ldots,\phi-1$. Accordingly, among the $MN-1$ sidelobes, $M\phi - 1$ of them attain ESL of $\left(\kappa - 1\right)MN$, while the remaining $M(N - \phi)$ sidelobes take the value $MN$, which leads directly to \eqref{EISL_FST_AFDM}.

\section{Proof of Theorem \ref{thm_OFDM_OTFS_OPT}}\label{proof_OFDM_OTFS_OPT}
 Let us revisit \eqref{the average sidelobe of stop-and-go model}. For sub-Gaussian ($\kappa < 2$) constellations, minimizing the ESL at each delay-Doppler bin over FST-AF is equivalent to solving the following optimization problem:
\begin{align}
    \max\limits_{\mathbf{U} \in \mathcal{U}(N),\mathbf{V} \in \mathcal{U}(M)}\;\;\left\| \left| \mathbf{V} \right|^2 \mathbf{f}_{M,q+1}^\ast \right\|^2 \left\| \left|\mathbf{U}^H \mathbf{F}_N^H \right|^2 \mathbf{f}_{N,k+1} \right\|^2.
\end{align}
Since both $\mathbf{V}$ and $\mathbf{U}^H \mathbf{F}_N^H$ are unitary matrices, it is clear that both $\left| \mathbf{V}\right|^2$ and $\left|\mathbf{U}^H \mathbf{F}_N^H \right|^2$ are bistochastic matrices, i.e., each of their rows and columns sums to 1. The optimality of OFDM then follows from majorization theory and the Schur-convexity of the $\ell_2$ norm.

Given two real vectors $\mathbf{a},\mathbf{b} \in \mathbb{R}^{N}$, let $\mathbf{a}^\downarrow$ denote the vector obtained by sorting the entries of $\mathbf{a}$ in descending order, i.e., $a_1^\downarrow \ge a_2^\downarrow \ge \dots \ge a_N^\downarrow$, and similarly for $\mathbf{b}^\downarrow$. If
\begin{align}
\sum_{n=1}^N a_n = \sum_{n=1}^N b_n,\quad
\sum_{n=1}^k a_n^\downarrow \geq \sum_{n=1}^k b_n^\downarrow,
\end{align}
we say that $\mathbf{a}$ majorizes $\mathbf{b}$, written as $\mathbf{a} \succ \mathbf{b}$. A necessary and sufficient condition for $\mathbf{a} \succ \mathbf{b}$ is that there exists a bistochastic matrix $\mathbf{K}$ such that $\mathbf{b} = \mathbf{K}\mathbf{a}$.

Next, note that
\begin{align}
    \left\| \left| \mathbf{V} \right|^2 \mathbf{f}_{M,q+1}^\ast \right\|^2 = \left\| \left| \mathbf{V} \right|^2 \mathbf{f}_{M,q+1}^R \right\|^2 + \left\| \left| \mathbf{V} \right|^2 \mathbf{f}_{M,q+1}^I \right\|^2,
\end{align}
where $\mathbf{f}_{M,q+1}^R$ and $\mathbf{f}_{M,q+1}^I$ are the real and imaginary parts of $\mathbf{f}_{M,q+1}$. It follows that
\begin{align}
     \left| \mathbf{V} \right|^2\mathbf{f}_{M,q+1}^R\prec\mathbf{f}_{M,q+1}^R,\;\;\left| \mathbf{V} \right|^2\mathbf{f}_{M,q+1}^I\prec\mathbf{f}_{M,q+1}^I.
\end{align}
Because the $\ell_2$ norm is Schur-convex and hence monotonically ordered with respect to majorization, it follows that
\begin{align}
    \left\| \left| \mathbf{V} \right|^2 \mathbf{f}_{M,q+1}^R \right\|^2 \le \left\|  \mathbf{f}_{M,q+1}^R \right\|^2,\left\| \left| \mathbf{V} \right|^2 \mathbf{f}_{M,q+1}^R \right\|^2 \le \left\|  \mathbf{f}_{M,q+1}^R \right\|^2,
\end{align}
and hence
\begin{equation}
     \left\| \left| \mathbf{V} \right|^2 \mathbf{f}_{M,q+1}^\ast \right\|^2 \le \left\|  \mathbf{f}_{M,q+1}^\ast \right\|^2,
\end{equation}
where the equality holds if $\mathbf{V} = \mathbf{I}_N$. By the same reasoning, we obtain
\begin{equation}
    \left\| \left|\mathbf{U}^H \mathbf{F}_N^H \right|^2 \mathbf{f}_{N,k+1} \right\|^2 \leq  \left\| \mathbf{f}_{N,k+1} \right\|^2,
\end{equation}
where the equality is achieved if $\mathbf{U} = \mathbf{F}_N^H$. This corresponds to the OFDM waveform.

For super-Gaussian ($\kappa > 2$) constellations, minimizing the ESL under the FST-AF is equivalent to
\begin{align}
    \min\limits_{\mathbf{U} \in \mathcal{U}(N),\mathbf{V} \in \mathcal{U}(M)}\;\;\left\| \left| \mathbf{V}^\ast \right|^2 \mathbf{f}_{M,q+1}^\ast \right\|^2 \left\| \left|\mathbf{U}^H \mathbf{F}_N^H \right|^2 \mathbf{f}_{N,k+1} \right\|^2.
\end{align}
To proceed, note that the matrix $\frac{1}{N} \mathbf{1}_N \mathbf{1}_N^T$ is a uniform bistochastic matrix. For any size-$N$ bistochastic matrix $\mathbf{K}$, it satisfies
\begin{equation}
    \frac{1}{N} \mathbf{1}_N \mathbf{1}_N^T\mathbf{K} = \frac{1}{N} \mathbf{1}_N \mathbf{1}_N^T.
\end{equation}
Thus,
\begin{align}
    &\frac{1}{M} \mathbf{1}_M \mathbf{1}_M^T\left| \mathbf{V} \right|^2 \mathbf{f}_{M,q+1}^\ast = \frac{1}{M} \mathbf{1}_M \mathbf{1}_M^T \mathbf{f}_{M,q+1}^\ast,\\
    &\frac{1}{N} \mathbf{1}_N \mathbf{1}_N^T\left|\mathbf{U}^H \mathbf{F}_N^H \right|^2 \mathbf{f}_{N,k+1} =  \frac{1}{N} \mathbf{1}_N \mathbf{1}_N^T\mathbf{f}_{N,k+1}.
\end{align}
Since the $\ell_2$ norm is Schur-convex, applying the majorization theory again yields
 \begin{align}
    & \left\| \left| \mathbf{V} \right|^2 \mathbf{f}_{M,q+1}^\ast \right\| ^2 \geq \left\| \frac{1}{M} \mathbf{1}_M \mathbf{1}_M^T \mathbf{f}_{M,q+1}^\ast \right\|^2, \\
    & \left\| \left|\mathbf{U}^H \mathbf{F}_N^H \right|^2 \mathbf{f}_{N,k+1} \right\| ^2 \geq \left\| \frac{1}{N} \mathbf{1}_N \mathbf{1}_N^T \mathbf{f}_{N,k+1} \right\|^2,
\end{align}
where the equalities are achieved if $\mathbf{V} = \mathbf{F}_M^H$ and $\mathbf{U} = \mathbf{I}_N$, corresponding to the OTFS waveform, completing the proof.

\bibliographystyle{IEEEtran}
\bibliography{reference_zy}

\end{document}